\pdfoutput=1
\documentclass[a4paper,UKenglish,cleveref, autoref, thm-restate]{lipics-v2021}
\usepackage[ruled]{algorithm}
\usepackage[noend]{algpseudocode}
\usepackage{stmaryrd}

\allowdisplaybreaks

\title{Semi-Streaming Algorithms for Submodular Function Maximization Under $b$-Matching, Matroid, and Matchoid Constraints} %TODO Please add

\titlerunning{Semi-Streaming Submodular Function Maximization Under $b$-Matching Constraint} %TODO optional, please use if title is longer than one line

\author{Chien-Chung Huang}{CNRS, DI ENS, École normale supérieure, Université PSL, France}{chien-chung.huang@ens.fr}{}{}%TODO mandatory, please use full name; only 1 author per \author macro; first two parameters are mandatory, other parameters can be empty. Please provide at least the name of the affiliation and the country. The full address is optional

\author{François Sellier}{École polytechnique, Institut Polytechnique de Paris, France}{francois.sellier@polytechnique.edu}{}{}

\authorrunning{C.-C. Huang and F. Sellier} %TODO mandatory. First: Use abbreviated first/middle names. Second (only in severe cases): Use first author plus 'et al.'

\Copyright{Chien-Chung Huang and François Sellier} %TODO mandatory, please use full first names. LIPIcs license is "CC-BY";  http://creativecommons.org/licenses/by/3.0/

 %TODO mandatory: Please choose ACM 2012 classifications from https://dl.acm.org/ccs/ccs_flat.cfm 

 %TODO mandatory; please add comma-separated list of keywords

\relatedversion{A preliminary version of this article appears in APPROX 2021~\cite{HS2021}.} %optional, e.g. full version hosted on arXiv, HAL, or other respository/website
\nolinenumbers %uncomment to disable line numbering

\hideLIPIcs  %uncomment to remove references to LIPIcs series (logo, DOI, ...), e.g. when preparing a pre-final version to be uploaded to arXiv or another public repository

\EventEditors{Mary Wootters and Laura Sanit\`{a}}
\EventNoEds{2}
\EventLongTitle{Approximation, Randomization, and Combinatorial Optimization. Algorithms and Techniques (APPROX/RANDOM 2021)}
\EventShortTitle{\mbox{\scriptsize APPROX/RANDOM 2021}}
\EventAcronym{APPROX/RANDOM}
\EventYear{2021}
\EventDate{August 16--18, 2021}
\EventLocation{University of Washington, Seattle, Washington, US (Virtual Conference)}
\EventLogo{}
\SeriesVolume{207}
\ArticleNo{14}
\begin{document}
\maketitle

\begin{abstract}
	We consider the problem of maximizing a non-negative submodular function under the $b$-matching constraint, in the semi-streaming model. 
	When the function is linear, monotone, 
	and non-monotone, we obtain the approximation ratios of  $2+\varepsilon$, $3 + 2 \sqrt{2} \approx 5.828$, and $4 + 2 \sqrt{3} \approx 7.464$, respectively. 
	
	We also consider a generalized problem, where a $k$-uniform 
	hypergraph is given, along with an extra matroid or a $k'$-matchoid constraint imposed on the edges, with the same goal of finding a $b$-matching that maximizes a submodular function. 
	When the extra constraint is a matroid, we obtain the approximation ratios of $k + 1 + \varepsilon$, $k + 2\sqrt{k+1} + 2$, and $k + 2\sqrt{k + 2} + 3$ for linear, monotone and non-monotone submodular functions, respectively. When the extra constraint is a $k'$-matchoid, we attain the approximation ratio $\frac{8}{3}k+ \frac{64}{9}k' + O(1)$ for general submodular functions.
\end{abstract}

\section{Introduction}

	%\subsection{Definitions and notations}

	Let $G=(V,E)$ be a multi-graph (without self-loops) where each vertex $v\in V$ is associated 
	with a \emph{capacity} $b_v \in \mathbb{Z}_{+}$. A \emph{$b$-matching} 
	is a subset of edges $M\subseteq E$ where each vertex $v$ has at most 
	$b_v$ incident edges contained in $M$. 
	
	In the maximum weight $b$-matching problem, edges are given weights 
	$w: E \rightarrow \mathbb{R}_{+}$ and we need to compute a $b$-matching 
	$M$ so that $w(M) = \sum_{e \in M}w(e)$ is maximized. 
	A generalization of the problem is that of 
	maximizing a non-negative \emph{submodular function} $f : 2^E \rightarrow \mathbb{R_+}$ 
	under the $b$-matching constraint, namely, we look for a $b$-matching 
	$M$ so that $f(M)$ is maximized. Here we recall the definition of a submodular function: \[\forall X \subseteq Y \subsetneq E, \forall e \in E\backslash Y, f(X \cup \{e\}) - f(X) \geq f(Y \cup \{e\}) - f(Y).\]
	Additionally $f$ is called \emph{monotone} if $\forall X \subseteq Y \subseteq E, f(X) \leq f(Y)$, otherwise it is \emph{non-monotone}. Observe that 
	the maximum weight matching is the special case where $f$ is a linear sum of the weights associated with the edges. Note that 
	throughout this paper, we implicitly assume that the submodular function $f$ is always non-negative. 
	
	In the traditional offline setting, 
	both problems are extensively studied. 
	The maximum weight $b$-matching can be solved in polynomial time~\cite{Sch2003}; 
	maximizing a non-negative monotone submodular function 
	under $b$-matching constraint 
	is NP-hard and the best approximation ratios so far are $2+\varepsilon$ and $4+\varepsilon$, 
	for the monotone and non-monotone case, respectively~\cite{Feldman2011}.  
	
	In this work, we consider the problem in 
	the semi-streaming model~\cite{Mut2005}. 
	Here the edges in $E$ arrive over time 
	but we have only limited space (ideally 
	proportional to the output size) and cannot 
	afford to store all edges in $E$---this rules out the possibility of applying known offline algorithms.
	
	We also consider a generalized problem, where a matroid or a matchoid 
	is imposed on the edges of an hypergraph. Specifically, here $G=(V,E)$ is a $k$-uniform hypergraph, where each edge $e \in E$ contains $k$ vertices in $V$. In addition to the capacities $b_v$, a matroid or a matchoid $\mathcal{M}$ (see the next section for the definitions) is imposed on the edges of $E$. A $b$-matching $M$ is feasible only if $M$ is independent in $\mathcal{M}$. The objective here is to find a feasible $b$-matching $M$ that maximizes $f(M)$ for $f$ a submodular function.
	
	%We recall that $\mathcal{M} = (E, \mathcal{I})$ is a matroid if the following three conditions hold: (1) $\emptyset \in \mathcal{I}$, (2) if $X\subseteq Y \in \mathcal{I}$, then 
	%$X\in \mathcal{I}$, and (3) if $X, Y \in \mathcal{I}, |Y| > |X|$, 
	%there exists an element $e \in Y \backslash X$ so that 
	%$X \cup \{e\} \in \mathcal{I}$. The sets in $\mathcal{I}$ are the \emph{independent sets} and the \emph{rank} $r_{\mathcal{M}}$ of the matroid $\mathcal{M}$ is defined as $\max_{X \in \mathcal{I}}|X|$. A minimum (inclusion-wise) non-independent subset is called a \emph{circuit}. For any subset $S \subseteq E$, its \emph{rank} $r_{\mathcal{M}}(S)$ in the matroid is defined as the maximum size of an independent set in $S$. We will also use the notion of \emph{span}, which is defined for a subset $S$ of $E$ as $\text{span}(S) = \{e \in E : r_{\mathcal{M}}(S) = r_{\mathcal{M}}(S \cup \{e\})$. Those definitions are close to the ones for independent families in vector spaces, and their nomenclature is similar.
	
	\subsection{Our Contribution} 
	
	We start with the maximum weight ($b$-)matching. For this problem, a long series of papers~\cite{CS2014,Epstein2009,FKMSZ2005,GW2019,LW2021,Mcg2005,PS2017,Zel2008} have proposed many
	semi-streaming algorithms, with progressively improved approximation ratios, culminating 
	in the work of Paz and Schwartzman~\cite{PS2017}, where 
	$2+\varepsilon$ approximation is attained, 
	for the simple matching. For the general 
	$b$-matching, very recently, Levin and Wajc~\cite{LW2021} gave a $3+\varepsilon$ approximation algorithm. We close the gap between the simple matching and the general $b$-matching.
	
	\begin{theorem} 
	\label{thm:first}
	    For the maximum weight $b$-matching problem, we obtain a $2+\varepsilon$ approximation algorithm using $O\left(\log_{1 + \varepsilon} (W/\varepsilon) \cdot |M_{max}|\right)$ variables in memory, and another using $O\left(\log_{1 + \varepsilon} (1/\varepsilon) \cdot |M_{max}| + \sum_{v \in V} b_v\right)$ variables, where $M_{max}$ denotes the maximum cardinality $b$-matching and $W$ denotes the maximum ratio between two non-zero weights.
	\end{theorem}
	
	The details of this theorem can be found in Section~\ref{sec:weight}.
	
	%Here we use ideas of Ghaffari and Wajc~\cite{GW2019} to guarantee the claimed space %requirement. 

	Next we consider the general case of submodular functions, for whom approximation 
	algorithms have been proposed in~\cite{DBLP:journals/mp/ChakrabartiK15,ChekuriGQ15,LW2021,DBLP:conf/nips/FeldmanK018}. The current best ratios obtained by Levin and Wajc~\cite{LW2021} are $3 + 2 \sqrt{2} \approx 5.828$ and $4 + 2 \sqrt{3} \approx 7.464$, for the monotone and 
	non-monotone functions, respectively.  We propose an alternative algorithm to achieve the same bounds. 
	
	\begin{theorem} 
	    \label{thm:second}
	    To maximize a non-negative submodular function under the $b$-matching constraint, we obtain algorithms providing a $3 + 2 \sqrt{2} \approx 5.828$ approximation for monotone functions and a $4 + 2 \sqrt{3} \approx 7.464$ approximation for non-monotone functions, using $O(\log W \cdot |M_{max}|)$ variables, where $M_{max}$ denotes the maximum cardinality $b$-matching and $W$ denotes the maximum quotient $\frac{f(e \,|\, Y)}{f(e' \,|\, X)}$, for $X \subseteq Y \subseteq E$, $e, e' \in E$, $f(e' \,|\, X) > 0$. 
	\end{theorem}
	
	The details of this theorem can be found in Section~\ref{sec:submodular}.
	
	It should be pointed out that in~\cite{LW2021}, for the case 
	of non-monotone functions, their algorithm only works for the simple matching, and it is unclear how to generalize it to the general 
	$b$-matching, while our algorithm lifts this restriction. Another interesting thing to observe is that even though the achieved ratios are the same and our analysis borrows
	ideas from~\cite{LW2021}, our algorithm is not really just 
	the same algorithm disguised under a different form. See Appendix~\ref{appendix_difference} for a concrete example where the two algorithms 
	behave differently. 
	
	\paragraph*{Extension: $k$-uniform hypergraph with a matroid or a matchoid constraint}
	
	%We also consider a generalization, where a $G=(V,E)$ is 
	%a $k$-uniform hypergraph
	
	Before we explain our results, let us recall some definitions. 
	
	A pair $\mathcal{M} = (E, \mathcal{I})$ is a matroid if the following three conditions hold: (1) $\emptyset \in \mathcal{I}$, (2) if $X\subseteq Y \in \mathcal{I}$, then 
	$X\in \mathcal{I}$, and (3) if $X, Y \in \mathcal{I}, |Y| > |X|$, 
	there exists an element $e \in Y \backslash X$ so that 
	$X \cup \{e\} \in \mathcal{I}$. The sets in $\mathcal{I}$ are the \emph{independent sets} and the \emph{rank} $r_{\mathcal{M}}$ of the matroid $\mathcal{M}$ is defined as $\max_{X \in \mathcal{I}}|X|$. A minimum (inclusion-wise) non-independent subset is called a \emph{circuit}. For any subset $S \subseteq E$, its \emph{rank} $r_{\mathcal{M}}(S)$ is defined as the maximum size of an independent set in $S$. We will also use the notion of \emph{span}, which is defined for a subset $S$ of $E$ as $\text{span}(S) = \{e \in E : r_{\mathcal{M}}(S) = r_{\mathcal{M}}(S \cup \{e\})$. These definitions are close to the ones for independent families in vector spaces, and their nomenclature is similar.
	
	In the following, we will sometimes refer to particular types of matroids: uniform and partition matroids. A \emph{uniform matroid} $\mathcal{M} = (E, \mathcal{I})$ of rank $r_{\mathcal{M}}$ is a matroid such that the independent sets are the sets of cardinality at most $r_{\mathcal{M}}$. A \emph{partition matroid} $\mathcal{M} = (E, \mathcal{I})$ is defined by a partition $E_1, \dots, E_l$ of $E$ and bounds $k_1,\dots,k_l$: a set $M \subseteq E$ is independent in $\mathcal{M}$ if for all $1 \leq i \leq l, |M \cap E_i| \leq k_i$.
	
	A matchoid $\mathcal{M} = (\mathcal{M}_i=(E_i,\mathcal{I}_i))_{i=1}^s$ is a collection of matroids $\mathcal{M}_i=(E_i,\mathcal{I}_i)$. Specifically, a $k'$-matchoid is a matchoid with the additional condition that each element $e \in E = \cup_{i=1}^{s}E_i$ appears in at most $k'$ of $E_i$s. A set $S \subseteq \cup_{i=1}^{s}E_i$ is independent in the matchoid $\mathcal{M}$ only if $S \cap E_i$ is independent for all $1 \leq i \leq s$. For instance, a partition matroid can be seen as $1$-matchoid made of uniform matroids.
	
	In our generalization, a $k$-uniform hypergraph is given 
	with an additional matroid or a $k'$-matchoid $\mathcal{M}$ imposed. 
	The goal is to find a $b$-matching $M$, which is also independent 
	in $\mathcal{M}$, so as to maximize $f(M)$.
	
	We notice that the problem of maximizing a submodular function under a $p$-matchoid constraint (without our additional $b$-matching constraint in the hypergraph) has already been studied by Chekuri et al.~\cite{ChekuriGQ15} and 
	Feldman et al.~\cite{DBLP:conf/nips/FeldmanK018}, where the best approximation ratios for the maximizing a monotone 
	and a non-monotone function are respectively $4p$ and $2p+2\sqrt{p(p+1)}+1$. It is easy to see that, if the constraint $\mathcal{M}$ is a matroid, our problem is a special case of the $(k+1)$-matchoid submodular function maximization, and if the constraint $\mathcal{M}$ is a $k'$-matchoid, our problem is a spacial case of the $(k + k')$-matchoid submodular function maximization\footnote{Indeed, one can observe that the capacity constraint associated with a vertex can be seen as a uniform matroid: each edge involves $k$ such uniform matroids.}.
	 
	Essentially in our next two theorems, 
	we show that if the majority of matroids contained in the matchoid are simple uniform matroids, we can attain better approximation ratios than in~\cite{ChekuriGQ15,DBLP:conf/nips/FeldmanK018}.

	% To see our problem is a $(k+1)$-matchoid, observe that we can define a uniform matroid on each vertex to replace the capacity constraint; as each edge appears in $k$ vertices, in total it appears in $k+1$ matroids.}. 
	%Using the current best algorithm of Chekuri et %al.~\cite{ChekuriGQ15} and 
	%Feldman et al.~\cite{DBLP:conf/nips/FeldmanK018}, one can 
	%obtain $4(k+1)$ and $2(k+1) + 2\sqrt{(k+1)(k+2)} + 1$ for monotone %and non-monotone functions, respectively. We obtain the following.
	
	First consider the case where $\mathcal{M}$ is a simple matroid.
	
	\begin{theorem} 
	    \label{thm:third}
        To maximize linear functions under the $b$-matching constraint with an additional matroid constraint, we obtain a $k+1+\varepsilon$ approximation algorithm, having a memory consumption of $O\left(\log_{1 + \varepsilon} (W/\varepsilon) \cdot \min\{|M_{max}|, r_{\mathcal{M}}\}\right)$ variables, where $M_{max}$ denotes the maximum cardinality $b$-matching and $W$ denotes the maximum ratio between two non-zero weights. To maximize non-negative submodular functions, we have an algorithm providing an approximation ratio of $k + 2\sqrt{k+1} + 1$ for monotone functions and $k + 2\sqrt{k + 2} + 3$ for non-monotone functions, using $O\left(\log_{1 + \varepsilon} (W/\varepsilon) \cdot \min\{|M_{max}|, r_{\mathcal{M}}\}\right)$ variables in memory, where $W$ denotes the maximum quotient $\frac{f(e \,|\, Y)}{f(e' \,|\, X)}$, for $X \subseteq Y \subseteq E$, $e, e' \in E$, such that $f(e' \,|\, X) > 0$. 
	\end{theorem}
	
	The details of this theorem can be found in Section~\ref{sec:matroid}.
	
	Observe that when $f$ is a monotone, or a non-monotone submodular function, our ratios are better than $4(k+1)$ and $2(k+1)+2\sqrt{(k+1)(k+2)} + 1$ achieved in~\cite{ChekuriGQ15} and in~\cite{DBLP:conf/nips/FeldmanK018}.  
	
	Our result is of particular interest  when $f$ is a linear function. In fact, notice that if the given hypergraph is $k$-partite, our problem is the same as finding the maximum weight intersection of $k$ partition matroids and an arbitrary matroid---this is a special case of finding a maximum weight intersection of $k+1$ arbitrary matroids, a problem recently studied by Garg et al.~\cite{GJS2021}. It can be seen that the set of stored elements in our algorithm is identical to that in the algorithm of~\cite{GJS2021}. Garg et al. have conjectured that there always exists a $k+1$ approximation in their stored elements (and have proved this to be true when $k+1=2$). Theorem~\ref{thm:third} thus confirms their conjecture for the special case when all but one matroids are partition matroids.
	
	We next consider the case where $\mathcal{M}$ is a $k'$-matchoid. 
	
	\begin{theorem} 
	\label{thm:fourth}
    To maximize a non-negative submodular function 
	under the $b$-matching constraint with an additional $k'$-matchoid constraint, we have an algorithm providing an approximation ratio better than $\frac{8}{3}k + \frac{64}{9}k'+O(1)$ for both monotone and non-monotone functions, using 
	$O(\log W \cdot |M_{max}|)$ variables in memory, where $M_{max}$ denotes the maximum cardinality $b$-matching and $W$ denotes the maximum quotient $\frac{f(e \,|\, Y)}{f(e' \,|\, X)}$, for $X \subseteq Y \subseteq E$, $e, e' \in E$, such that $f(e' \,|\, X) > 0$. 
	\end{theorem}
	
	The details of this theorem can be found in Section~\ref{sec:matchoid}.
	
	The exact expressions for the approximation ratios are formally stated in Theorems~\ref{theo_submo_mono_matchoid} and~\ref{theo_submo_nonmono_matchoid}. For a comparison, recall that using the techniques of~\cite{ChekuriGQ15,DBLP:conf/nips/FeldmanK018}, the obtained ratio is $4(k+k')$ and $2(k+k')+2\sqrt{(k+k'+1)(k+k'+2)}+1$ when $f$ is monotone and non-monotone, respectively. Our result shows that if $k$ is at least roughly three times of $k'$, better approximation ratios can be achieved.
	
	\paragraph*{Comparison with the conference version of this work.}
	    Compared with our previous results of~\cite{HS2021}, Theorems~\ref{thm:third} and~\ref{thm:fourth} are significant improvements. In fact, Theorem~\ref{thm:third} allows us to get approximation ratios of $k+1+\varepsilon$, $k + 2\sqrt{k+1} + 1$ and $k + 2\sqrt{k + 2} + 3$ for linear, monotone, and non-monotone submodular functions, respectively, meanwhile in~\cite{HS2021} the approximation ratios are $k+2\sqrt{k} + 1$, $\frac{8}{3}k + O(1)$, and $\frac{8}{3}k + O(1)$. Moreover, Theorem~\ref{thm:fourth} is a generalization of the technique used for the matroids in~\cite{HS2021} to matchoids.

	\subsection{Our Technique} 
	
	We use a local-ratio technique to decide 
	to retain or to discard a newly arrived edge during the streaming 
	phase. After this phase, a greedy algorithm, according to 
	the reverse edge arrival order, is then applied to add edges 
	one by one into the solution while guaranteeing feasibility. 
	
	This is in fact the same framework used in~\cite{LW2021,PS2017}. Our main technical innovation is to introduce a data structure, which takes the form of 
	a set of queues. Such a set of queues is associated with each vertex (and with the imposed matroid). Every edge, if retained, will appear as an element\footnote{In this article, we will often use ``edge'' and ``element'' interchangeably.} in one of these queues for each of its endpoints (and for the imposed matroid). These queues will guide the greedy algorithm to make better choices and are critical in our improvement over previous results. 
	%These queues serve two purposes: (1) they encode more %fine-tuned ``decision variables'' in each vertex, and (2) they %will guide the greedy algorithm to make better decisions. 
	Here we give some intuition behind these queues. Consider the maximum weight $b$-matching problem. Similar to~\cite{GJS2021}, we compute, for every edge, a \emph{gain}. The sum of the gains of the retained edges can be shown to be at least half of the real weight of the unknown optimal matching. The question then boils down to how to ``extract'' a matching whose real weight is large compared to these gains of the retained edges. In our queues, the elements are stacked in such a way that the weight of an element $e$ is the sum of the gains of all the elements preceding $e$ in any queue containing $e$. This suggests that if $e$ is taken by the greedy algorithm, we can as well ignore all elements that are underneath $e$ in the queues, as their gains are already ``paid for'' by $e$.

\section{Maximum Weight $b$-Matching}
\label{sec:weight}
    \subsection{Description of the Algorithm}
    
        For ease of description, we explain how to achieve $2$ approximation, ignoring the issue of space complexity for the moment. We will explain how a slight modification can ensure the desired space complexity, at the expense of an extra $\varepsilon$ term in the approximation ratio (see Appendix~\ref{appendix_memory}). 
    
        The formal algorithm for the streaming phase is shown in Algorithm~\ref{streaming_part}. We give an informal description here. Let $S$, initially empty, be the set of edges that have been stored so far. For each vertex $v \in V$, a set $Q_v=\{Q_{v,1},\cdots, Q_{v,b_v}\}$ of queues are maintained. These queues contain the edges incident to $v$ that are stored in $S$ and respect the arrival order of edges (newer edges are higher up in the queues). Each time a new edge $e$ arrives, we compute its \emph{gain} $g(e)$ (see Lines 5 and 8). Edge $e$ is added into $S$ only if its gain is strictly positive. If this is the case, for each endpoint $u$ of $e$, we put $e$ in one of $u$'s queues (see Lines 6 and 13) and define a \emph{reduced weight} $w_u(e)$ (Line 11). It should be noted that $w_u(e)$ will be exactly the sum of the gains of the edges preceding (and including) $e$ in the queue. We refer to the last element inserted in a queue $Q$ as the top element of that queue, denoted $Q.top()$. To insert an element $e$ on top of a queue $Q$, we use the instruction $Q.push(e)$.  By convention, for an empty queue $Q$ we have $Q.top() = \bot$. We also set $w_u(\bot) = 0$. Notice that each element $e$ also has, for each endpoint $v \in e$, a pointer $r_v(e)$ to indicate its immediate 
        predecessor in the queue of $v$, where it appears.

    	\begin{algorithm}
    	\caption{Streaming phase for weighted matching}\label{streaming_part}
    	\begin{algorithmic}[1]
    	\State $S \gets \emptyset$
    	\State $\forall v \in V : Q_v \gets (Q_{v,1} = \emptyset, \cdots, Q_{v, b_v} = \emptyset)$ \Comment $b_v$ queues for a vertex $v$ 
    	\For{$e = e_t,\, 1 \leq t \leq |E|$ an edge from the stream}
    	
    		\For{$u \in e$}
    			\State $w_u^*(e) \gets \min \{w_u(Q_{u, q}.top()) : 1 \leq q \leq b_u\}$
    			\State $q_u(e) \gets q \text{ such that $w_u(Q_{u, q}.top()) = w_u^*(e)$}$
    		\EndFor
    	
    		\If {$w(e) > \sum_{u \in e} w_u^*(e)$}
    			\State $g(e) \gets w(e) - \sum_{u \in e} w_u^*(e)$
    			\State $S \gets S \cup \{e\}$
    			\For{$u \in e$}
    				\State $w_u(e) \gets w_u^*(e) + g(e)$
    				\State $r_u(e) \gets Q_{u, q_u(e)}.top()$ \Comment $r_u(e)$ is the element below $e$ in the queue
    				\State $Q_{u, q_u(e)}.push(e)$ \Comment add $e$ on the top of the smallest queue
    			\EndFor
    		\EndIf
    	\EndFor
    	\end{algorithmic}
    	\end{algorithm}
    		
    	After the streaming phase, our greedy algorithm, formally described in Algorithm~\ref{greedy_part}, constructs a $b$-matching based on the stored set $S$. 
    	
    	The greedy proceeds based on the reverse edge arrival order---but with an important modification. Once an edge $e$ is taken as part of the $b$-matching, all edges preceding $e$ that are stored in the same queue as $e$ will be subsequently ignored by the greedy algorithm. The variables $z_e$ are used to mark this fact. 
    		
    	\begin{algorithm}
    	\caption{Greedy construction phase}\label{greedy_part}
    	\begin{algorithmic}[1]
    	\State $M \gets \emptyset$
    	\State $\forall e \in S : z_e \gets 1$
    	\For{$e \in S$ in reverse order}
    	
    		\If{$z_e = 0$} \textbf{continue} \Comment skip edge $e$ if it is marked
    		\EndIf
    		
    		\State $M \gets M \cup \{e\}$
    	
    		\For{$u \in e$}
    			\State $c \gets e$
    			\While{$c \neq \bot$}
    				\State $z_c \gets 0$ \Comment mark elements below $e$ in each queue
    				\State $c \gets r_u(c)$
    			\EndWhile
    		\EndFor
    		
    	\EndFor
    	\State \Return $M$
    	\end{algorithmic}
    	\end{algorithm}
	
	\subsection{Analysis for Maximum Weight $b$-Matching}
	
	    For analysis, for each discarded element $e \in E \backslash S$, we set $g(e) = 0$ and $w_u(e) = w_u^*(e)$ for each $u \in e$. The \emph{weight} of a queue, $w_u(Q_{u,i})$, is defined as the reduced weight of its top element, namely, $w_u(Q_{u, i}.top())$. Let $w_u(Q_u)=\sum_{i=1}^{b_u}w_u(Q_{u,i})$. We write $S^{(t)}$ as the value of $S$ at the end of the iteration $t$ of the streaming phase, and by convention $S^{(0)} = \emptyset$. This notation $^{(t)}$ will also be used for other sets such as $Q_u$ and $Q_{u,i}$. Through this paper, $M^{opt}$ will always refer to the best solution for the considered problem.

        The following proposition follows easily by induction. 
		\begin{proposition} \label{lemma_g_w}
    		\begin{enumerate}[(i)]
    			\item For all $v \in V$ we have $g(\delta(v)) = g(\delta(v) \cap S) = w_v(Q_v)$. 
    			\item The set $\{Q_{v, q}.top() : 1 \leq q \leq b_v\}$ contains the $b_v$ heaviest elements of $S \cap \delta(v)$ in terms of reduced weights.
    		\end{enumerate}	
		\end{proposition}
		
		\begin{lemma} \label{lemma_w_T_w_opt}
		    At the end of Algorithm~\ref{streaming_part}, for all $b$-matching $M'$ and for all $v \in V$, we have $w_v(Q_v) \geq w_v(M' \cap \delta(v))$.
		\end{lemma}
		
		\begin{proof}
		    By Proposition~\ref{lemma_g_w}(ii), $w_v(Q_v)$ is exactly the sum of the reduced weights of the $b_v$ heaviest elements in $S \cap \delta(v)$ (which are on top of the queues of $Q_v$). If we can show that for each element $e = e_t \in M' \backslash S$, $w_u(e_t) \leq \min \{w_v(Q^{|E|}_{v, q}) : 1 \leq q \leq b_v\}$, the proof will follow. Indeed, as $e_t$ is discarded, we know that $w_v(e_t) =  \min \{w_v(Q_{v, q}^{(t-1)}) : 1 \leq q \leq b_v\} \leq \min \{w_v(Q^{|E|}_{v, q}) : 1 \leq q \leq b_v\}$, where the inequality holds because the weight of a queue is monotonically increasing. 
		\end{proof}
		
		\begin{lemma} \label{lemma_g_ineg}
		$2 g(S) \geq w(M^{opt})$.
		\end{lemma}
		
		\begin{proof}
			It is clear that for $e = \{u, v\}$ we have $w_u(e) + w_v(e) \geq w(e)$. Therefore
        	\begin{eqnarray*}		
            	w(M^{opt}) &\leq & \sum_{e = \{u, v\} \in M^{opt}}w_u(e) + w_v(e) 
            		= \sum_{u \in V} w_u(M^{opt} \cap \delta(u)) \\
            		&\leq & \sum_{u \in V} w_u(Q_u) 
            	= \sum_{u \in V} g(S \cap \delta(u))  
            	= 2 g(S),
        	\end{eqnarray*}
	        where the second inequality follows from Lemma~\ref{lemma_w_T_w_opt} and the subsequent equality from Proposition~\ref{lemma_g_w}(i). The last equality comes from the fact that an edge is incident to $2$ vertices.
		\end{proof}
		
		Recall that $q_v(e)$ refers to the index of the particular queue in $Q_v$ where a new edge $e$ will be inserted (Line~6 of Algorithm~\ref{streaming_part}). 
		
		\begin{lemma} \label{lemma_greedy_weight}
			Algorithm~\ref{greedy_part} outputs a feasible $b$-matching $M$ with weight $w(M) \geq g(S)$.
		\end{lemma}
		
		\begin{proof}
			By an easy induction, we know that for a given $e = e_t \in S$ and $v \in e$, we have:
			\begin{equation}
			    w_v(e) = w_v^*(e) + g(e) = \sum_{e' \in Q_{v, q_v(e)}^{(t)}} g(e') \,\text{ and }\, w(e) = g(e) + \sum_{u \in e} \sum_{e' \in Q_{u, q_u(e)}^{(t-1)}}g(e'). 
			    \label{equ:charged}
			\end{equation}
			
			Moreover, observe that $S \cap \delta(v)$ can be written as a disjoint union of the $Q_{v, q}$ for $1 \leq q \leq b_v$: $S \cap \delta(v) = \bigcup_{1 \leq q \leq b_v} Q_{v, q}$. 
			One can also observe that Algorithm~\ref{greedy_part} takes at most one element in each queue $Q_{v, i}$. In fact, an element can be added only if no element above it in any of the queues where it appears  has already been added into $M$; and no element below it in the queues can be already part of $M$ because $S$ is read in the reverse arrival order. Consequently $M$ respects the capacity constraint and is thus a feasible $b$-matching. We now make a critical claim from which the correctness of the lemma follows easily. 
			
			\begin{claim} 
    			Given an edge $e \in S$, either $e \in M$, or there exists another edge $e'$ arriving later than $e$, such that $e' \in M$ and there exists a queue belonging to a common endpoint of $e$ and $e'$, which contains both of them.
			\end{claim}
			
			Observe that if the claim holds, by~(\ref{equ:charged}), the gain $g(e)$ of any edge $e \in S$ will be ``paid'' for by some edge $e' \in M$ and the proof will follow. 
			
			To prove the claim, let $e = \{u, v\}$ and assume that $e \not \in M$. Consider the two queues $Q_{u, q_u(e)}$ and $Q_{v, q_v(e)}$. The edges stored above $e$ in these two queues must have arrived later than $e$ in $S$ and have thus already been considered by Algorithm~\ref{greedy_part}. The only reason that $e \not \in M$ must be that $z_e=0$ when $e$ is processed, implying that one of these edges was already part of $M$. Hence the claim follows.
		\end{proof}
		
		Lemmas~\ref{lemma_g_ineg} and~\ref{lemma_greedy_weight} give the following theorem:
		\begin{theorem}
			Algorithms \ref{streaming_part} and \ref{greedy_part} provide a $2$ approximation for the maximum weight $b$-matching problem.
		\end{theorem}

		We refer the readers to Appendix~\ref{appendix_memory} for the details on how to handle the memory consumption of the algorithm. 
		
		\begin{remark}
		    It is straightforward to extend our algorithm to a $k$-uniform hypergraph, where we can get an approximation ratio of $k$. Notice that if the $k$-uniform hypergraph is also $k$-partite, then the problem becomes that of finding a maximum weight intersection of $k$ partition matroids. It can be shown that our stored edge set is exactly identical to the one stored by the algorithm of Garg et al.~\cite{GJS2021}. They have conjectured that for $k$ arbitrary general matroids, their stored edge set always contains a $k$ approximation. Our result thus proves their conjecture to be true when all matroids are partition matroids.
		\end{remark}
		
\section{Submodular Function Maximization}
\label{sec:submodular}
	\subsection{Description of the Algorithm}
	
	    For submodular function maximization, the streaming algorithm, formally described in Algorithm~\ref{streaming_part_submod}, is quite similar to the one for the weighted $b$-matching in the preceding section. Here notice that the element weight $w(e)$ is replaced by the marginal value $f(e \,|\, S)$ (see Lines 7 and 10). We use a similar randomization method to that of Levin and Wajc~\cite{LW2021} for non-monotone functions (adding an element to $S$ only with probability $p$, see Lines~8-9), and our analysis will bear much similarity to theirs. The greedy algorithm to build a solution $M$ from $S$ is still Algorithm~\ref{greedy_part}.

		\begin{algorithm}
		\caption{Streaming phase for submodular function maximization}\label{streaming_part_submod}
		\begin{algorithmic}[1]
		\State $S\gets \emptyset$
		\State $\forall v \in V : Q_v \gets (Q_{v,1} = \emptyset, \cdots, Q_{v, b_v} = \emptyset)$
		\For{$e = e_t,\, 1 \leq t \leq |E|$ an edge from the stream}
			\For{$u \in e$}
				\State $w_u^*(e) \gets \min \{w_u(Q_{u, q}.top()) : 1 \leq q \leq b_u\}$
				\State $q_u(e) \gets q \text{ such that $w_u(Q_{u, q}.top()) = w_u^*(e)$}$
			\EndFor
		
			\If {$f(e\,|\,S) > \alpha \sum_{u \in e} w_u^*(e)$}
				\State $\pi \gets \text{a random variable equal to $1$ with probability $p$ and $0$ otherwise}$
				\If{$\pi = 0$} \textbf{continue} \Comment skip edge $e$ with probability $1 - p$
				\EndIf
				\State $g(e) \gets f(e\,|\,S) - \sum_{u \in e} w_u^*(e)$
				\State $S \gets S \cup \{e\}$
				\For{$u \in e$}
					\State $w_u(e) \gets w_u^*(e) + g(e)$
					\State $r_u(e) \gets Q_{u, q_u(e)}.top()$
					\State $Q_{u, q_u(e)}.push(e)$
				\EndFor
			\EndIf
		\EndFor
		\end{algorithmic}
		\end{algorithm}
		
		Algorithm~\ref{streaming_part_submod} uses, for $\alpha = 1 + \varepsilon$, $O(\log_{1 + \varepsilon} (W / \varepsilon) \cdot |M_{max}|)$ variables, where $M_{max}$ denotes the maximum cardinality $b$-matching and $W$ denotes the maximum quotient $\frac{f(e \,|\, Y)}{f(e' \,|\, X)}$, for $X \subseteq Y \subseteq E$, $e, e' \in E$, $f(e' \,|\, X) > 0$ (in Appendix~\ref{appendix_memory} we explain how to guarantee such space complexity when $f$ is 
		linear---the general case of a submodular function
		follows similar ideas). 
	
	\subsection{Analysis for Monotone Submodular Function Maximization}
	
		Let $\alpha = 1 + \varepsilon$. In this section, $p = 1$ (so we have actually a deterministic 
		algorithm for the monotone case). The following two lemmas relate the total gain 
	    $g(S)$ with the marginal values $f(S\,|\,\emptyset)$ and $f(M^{opt} \,|\, S)$.
		\begin{lemma} \label{lemma_g_f_ineg}
			It holds that $g(S) \geq \frac{\varepsilon}{1 + \varepsilon}f(S\,|\,\emptyset)$.
		\end{lemma}
		
		\begin{proof}
			For an element $e = e_t \in S$ we have $f\left(e\,|\,S^{(t-1)}\right) \geq (1 + \varepsilon) \sum_{u \in e}w_u^*(e)$ so
			\[g(e) = f\left(e \,|\, S^{(t-1)}\right) - \sum_{u \in e} w_u^*(e) \geq f\left(e \,|\, S^{(t-1)}\right)\left(1 - \frac{1}{1 + \varepsilon}\right),\]
			implying that 
			\[g(S) = \sum_{e \in S} g(e) \geq \sum_{e = e_t \in S} f\left(e \,|\, S^{(t-1)}\right)\left(1 - \frac{1}{1 + \varepsilon}\right) = \frac{\varepsilon}{1 + \varepsilon} f(S \,|\, \emptyset).\]
		\end{proof}
		
		As in the previous section, if an edge $e$ is discarded, we assume that $w^*_v(e)=w_v(e)$ 
		for each $v \in e$. 
		
		\begin{lemma} \label{lemma_2g_geq_f_MS}
		It holds that $2 (1 + \varepsilon) g(S) \geq f(M^{opt} \,|\, S)$.
		\end{lemma}
		
		\begin{proof}
			The only elements $e = e_t$ missing in $S$ are the ones satisfying the inequality $f(e \,|\, S^{(t-1)}) \leq (1 + \varepsilon) \sum_{u \in e} w_u^*(e)$. So by submodularity, 
			\begin{align*}
				f(M^{opt} \,|\, S) &\leq \sum_{e \in M^{opt} \backslash S} f(e \,|\, S)
				\leq \sum_{e = e_t \in M^{opt} \backslash S} f(e \,|\, S^{(t-1)})\\
				&\leq \sum_{e \in M^{opt} \backslash S} (1 + \varepsilon) \sum_{u \in e} w_u^*(e)
				= (1 + \varepsilon)\sum_{e \in M^{opt} \backslash S} \sum_{u \in e} w_u(e)\\
				&= (1 + \varepsilon) \sum_{u \in V} w_u((M^{opt} \backslash S) \cap \delta(u))
				\leq (1 + \varepsilon) \sum_{u \in V} w_u(Q_u) \\
				&\leq 2(1 + \varepsilon) g(S), 
			\end{align*}
			similar to the proof of  Lemma~\ref{lemma_g_ineg}.
		\end{proof}
		
		\begin{lemma} \label{lemma_greedy_submod}
			Algorithm~\ref{greedy_part} outputs a feasible $b$-matching with $f(M)\geq g(S) + f(\emptyset)$.
		\end{lemma}
		
		\begin{proof}
		   As argued in the proof of 
		   Lemma~\ref{lemma_greedy_weight}, $M$ respects the capacities and so is feasible. Now, suppose that $M = \{e_{t_1}, \cdots, e_{t_{|M|}}\}$ , $t_1 < \cdots < t_{|M|}$. Then 
			\begin{align*}
			    f(M) = f(\emptyset) + \sum_{i = 1}^{|M|}f(e_{t_i} \,|\, \{e_{t_1}, \cdots, e_{t_{i-1}}\}) &\geq f(\emptyset) + \sum_{i = 1}^{|M|}f(e_{t_i} \,|\, S^{(t_{i}-1)})\geq f(\emptyset) + g(S),
			\end{align*}
			as the values $f(e_{t_i} \,|\, S^{(t_{i-1})})$ play the same role as the weights in Lemma~\ref{lemma_greedy_weight}.
		\end{proof}
		
		\begin{theorem} \label{theo_submod_mono}
			Algorithms~\ref{streaming_part_submod} and~\ref{greedy_part} provide a $3 + 2 \sqrt{2}$ approximation if we set $\varepsilon = \frac{1}{\sqrt{2}}$.
		\end{theorem}
		
		\begin{proof}
		By Lemmas~\ref{lemma_g_f_ineg} and~\ref{lemma_2g_geq_f_MS}, 
			we derive $\left(2 + 2 \varepsilon + \frac{1 + \varepsilon}{\varepsilon}\right)g(S) \geq f(M^{opt} \,|\, S) + f(S \,|\, \emptyset) = f(M^{opt} \cup S \,|\, \emptyset) \geq f(M^{opt} \,|\, \emptyset)$, where the last inequality is due to the monotonicity of $f$. 
			By Lemma~\ref{lemma_greedy_submod}, the 
			output $b$-matching $M$ guarantees that $f(M) \geq g(S) + f(\emptyset)$. As a result, $\left(3 + 2 \varepsilon + \frac{1}{\varepsilon}\right)f(M) \geq f(M^{opt} \,|\, \emptyset) + f(\emptyset) = f(M^{opt})$. Setting  $\varepsilon = \frac{1}{\sqrt{2}}$ gives the result. 
			%gives a ratio of $3 + 2 \sqrt{2} \approx %5.828$.
		\end{proof}
		
		\begin{remark}
		    When $b_v=1$ for all $v\in V$ (\emph{i.e.} simple matching), our algorithm behaves exactly the same as the algorithm of Levin and Wajc~\cite{LW2021}. Therefore their tight example also applies to our algorithm. In other words, our analysis of approximation ratio is tight.
		\end{remark}
		
	\subsection{Analysis for Non-Monotone Submodular Function Maximization}
	
		In this section, we suppose that $\frac{1}{3 + 2 \varepsilon} \leq p \leq \frac{1}{2}$. 
		
		\begin{lemma} \label{lemma_eg_ef}
		It holds that 
			\[\left( 2(1 + \varepsilon) + \frac{1 + \varepsilon}{\varepsilon}\right)\mathbb{E}[g(S)] \geq \mathbb{E}[f(S \cup M^{opt} \,|\, \emptyset)].\]
		\end{lemma}
	
		\begin{proof}
			From Lemma~\ref{lemma_g_f_ineg} we have that for any execution of the algorithm (a realization of randomness), the inequality $\frac{1 + \varepsilon}{\varepsilon}g(S) \geq f(S \,|\, \emptyset)$ holds, so it is also true in expectation. We will try to prove in the following that $2(1 + \varepsilon)\mathbb{E}[g(S)] \geq \mathbb{E}[f(M^{opt} \,|\, S)]$, which is the counterpart of Lemma~\ref{lemma_2g_geq_f_MS}.
			
			First, we show that for any $e \in M^{opt}$:
			\begin{equation} \label{ineg_fin}
			(1 + \varepsilon)\mathbb{E}\left[\sum_{u \in e} w_u(e)\right] \geq \mathbb{E}[f(e \,|\, S)]
			\end{equation}
			We will use a conditioning similar to the one used in~\cite{LW2021}. Let $e = e_t \in M^{opt}$. We consider the event $A_e = [f(e \,|\, S^{(t-1)}) \leq (1 + \varepsilon)\sum_{u \in e}w_u^*(e)]$. Notice that if $A_e$ holds, $e$ is not part of $S$ and $w^*_v(e)=w_v(e)$ for each $v \in e$. Now by submodularity, 
			\begin{align*}
				\mathbb{E}[f(e \,|\, S)\,|\,A_e] \leq \mathbb{E}[f(e \,|\, S^{(t-1)})\,|\,A_e] &\leq \mathbb{E}\left[(1 + \varepsilon)\sum_{u \in e}w_u^*(e)\,|\,A_e\right]\\
				&= (1 + \varepsilon)\mathbb{E}\left[\sum_{u \in e} w_u(e)\,|\,A_e\right]
			\end{align*}
			Next we consider the condition $\overline{A_e}$ (where the edge $e$ should be added into $S$ with probability $p$). As $p \leq \frac{1}{2}$, and for $e = e_t = \{u, v\}$ we have the inequality $w_u(e) + w_v(e) = 2 f(e\,|\,S^{(t-1)}) - w_u^*(e) - w_v^*(e)$ when $e$ is added to $S$, we get
			\begin{align*}
				\mathbb{E}\left[\sum_{u \in e}w_u(e)\,|\,\overline{A_e}\right] &= p \cdot \mathbb{E}\left[2 f(e \,|\, S^{(t-1)}) - \sum_{u \in e} w_u^*(e)\,|\,\overline{A_e}\right] + (1 - p)\cdot\mathbb{E}\left[\sum_{u \in e} w_u^*(e)\,|\,\overline{A_e}\right]\\
				&= 2p \cdot \mathbb{E}\left[f(e \,|\, S^{(t-1)})\,|\,\overline{A_e}\right] + (1 - 2p)\cdot\mathbb{E}\left[\sum_{u \in e} w_u^*(e)\,|\,\overline{A_e}\right]\\
				&\geq 2p \cdot \mathbb{E}\left[f(e \,|\, S^{(t-1)})\,|\,\overline{A_e}\right].
			\end{align*}
			As a result, for $p \geq \frac{1}{3 + 2 \varepsilon}$, 
			\begin{align*}
			(1 + \varepsilon)\mathbb{E}\left[\sum_{u \in e} w_u(e)\,|\,\overline{A_e}\right]
				&\geq 2p(1 + \varepsilon)\cdot \mathbb{E}\left[f(e \,|\, S^{(t-1)})\,|\,\overline{A_e}\right]\\
				&\geq (1 - p)\cdot \mathbb{E}\left[f(e \,|\, S^{(t-1)})\,|\,\overline{A_e}\right]\\
				&\geq \mathbb{E}\left[f(e \,|\, S)\,|\,\overline{A_e}\right],
			\end{align*}
			where the last inequality holds  
			because with probability $p$ we have $f(e \,|\, S) = 0$ 
			(as $e \in S$) 
			and with probability $1-p$, $f(e \,|\, S) \leq f(e \,|\, S^{(t-1)})$ (by submodularity). 
			
			So we have proven inequality (\ref{ineg_fin}) and it follows that 
			\begin{align*}
				\mathbb{E}\left[f(M^{opt} \,|\, S)\right] &\leq \sum_{e \in M^{opt}} \mathbb{E}\left[f(e \,|\, S)\right]
				\leq (1 + \varepsilon) \sum_{e \in M^{opt}} \mathbb{E}\left[\sum_{u \in e} w_u(e)\right]\\	
                &= (1 + \varepsilon) \sum_{u \in V} \sum_{e \in M^{opt} \cap \delta(u)} \mathbb{E}\left[w_u(e)\right]
                \leq (1 + \varepsilon) \sum_{u \in V} \mathbb{E}\left[w_u(Q_u)\right]\\
&=2(1 + \varepsilon)\mathbb{E}[g(S)],  
			\end{align*}
		where in the last inequality we use the fact that Lemma~\ref{lemma_w_T_w_opt} holds for every realization of randomness. 
			
		Now the bounds on $\mathbb{E}\left[f(M^{opt} \,|\, S)\right]$ and the bound on $\mathbb{E}\left[f(S \,|\, \emptyset)\right]$ argued in the beginning give the proof 
		of the lemma.
		\end{proof}
		
		Then we will use the following lemma, due to due to Buchbinder et al.~\cite{Buchbinder2014a}:
		
		\begin{lemma}[Lemma 2.2 in~\cite{Buchbinder2014a}] \label{lemma_ineg_h}
			Let $h : 2^N \rightarrow \mathbb{R}_+$ be a non-negative submodular function, and let $B$ be a
random subset of $N$ containing every element of $N$ with probability at most $p$ (not necessarily
independently), then $\mathbb{E}[h(B)] \geq (1 - p) h(\emptyset)$.
		\end{lemma}
		
		\begin{theorem}
			Algorithm~\ref{streaming_part_submod} run with $p = \frac{1}{3 + 2\varepsilon}$ provides a set $S$, upon which Algorithm~\ref{greedy_part} outputs a $b$-matching $M$ satisfying:
			\[\left(\frac{4 \varepsilon^2 + 8\varepsilon + 3}{2 \varepsilon}\right) \mathbb{E}[f(M)] \geq f(M^{opt}).\]
			This ratio is optimized when $\varepsilon = \frac{\sqrt{3}}{2}$, which gives a $4 + 2\sqrt{3}$ approximation.
		\end{theorem}
		
		\begin{proof}
			Combining Lemma~\ref{lemma_greedy_submod} and Lemma~\ref{lemma_eg_ef}, 
			\[\left( 2(1 + \varepsilon) + \frac{1 + \varepsilon}{\varepsilon}\right)\mathbb{E}[f(M)] \geq \mathbb{E}[f(S \cup M^{opt})].\]
			
			Now we can apply Lemma~\ref{lemma_ineg_h} by defining
			$h : 2^E \rightarrow \mathbb{R}_+$ as, for any $X \subseteq E$, $h(X) = f(X \cup M^{opt})$ (trivially $h$ is non-negative and submodular). As any element of $E$ has the probability of at most $p$ to appear in $S$, we derive $\mathbb{E}[f(S \cup M^{opt})] = \mathbb{E}[h(S)] \geq (1 - p) h(\emptyset) = (1-p) f(M^{opt})$. Therefore,
			\[\left( 3 + 2 \varepsilon + \frac{1}{\varepsilon}\right)\mathbb{E}[f(M)] \geq \mathbb{E}[f(S \cup M^{opt})] \geq (1-p) f(M^{opt}).\]
			As $p = \frac{1}{3 + 2\epsilon}$, we have
			\[\left(\frac{4 \varepsilon^2 + 8\varepsilon + 3}{2 \varepsilon}\right)\cdot \mathbb{E}[f(M)] \geq f(M^{opt}).\]
			This ratio is optimized when $\varepsilon = \frac{\sqrt{3}}{2}$, which gives a $4 + 2\sqrt{3} \approx 7.464$ approximation.
		\end{proof}

\section{Matroid-constrained Maximum Submodular $b$-Matching}
\label{sec:matroid}
    In this section we consider the more general case of a $b$-matching on a $k$-uniform hypergraph and 
    we impose a matroid constraint $\mathcal{M} = (E, \mathcal{I})$. 
    A matching $M \subseteq E$ is feasible only if it respects 
    the capacities of the vertices and is an independent set in the matroid $\mathcal{M}$.

	\subsection{Description of the Algorithm}
	
    	For the streaming phase, our algorithm, formally described in Algorithm~\ref{streaming_part_matroid_submod}, is a generalization of Algorithm~\ref{streaming_part_submod} in the last section. We let $\alpha = 1 + \varepsilon \geq 1$. For the matroid $\mathcal{M}$, we maintain a set of queues $Q_{\mathcal{M}}=\{Q_{\mathcal{M},1},\cdots, Q_{\mathcal{M},r_{\mathcal{M}}}\}$, where $r_{\mathcal{M}}$ is the rank of $\mathcal{M}$, to store the elements of $S$ (so if an edge $e$ is part of $S$, it appears in a total of $k+1$ queues, $k$ of them corresponding to the vertices in $e$, and the remaining one corresponding to the matroid). To facilitate the presentation, we write $Top(Q_{\mathcal{M}})$ to denote the set of the elements on top of the queues of $Q_{\mathcal{M}}$. Lines 8-13 will guarantee that $Top(Q_{\mathcal{M}})$ is an independent set at any time (in fact a maximum weight independent set among all elements arrived so far, according to the reduced weights---see Lemma~\ref{lem:maximum_weight_base}). Notice that in the 
    	end of the algorithm,
    	$Top(Q_{\mathcal{M}})$ is not always fully-ranked (\emph{i.e.} a base 
    	of $\mathcal{M}=(E, \mathcal{I}$), but it is always fully-ranked 
    	in the restricted matroid $\mathcal{M}|S=(S, \mathcal{I}_S)$, 
    	where $\mathcal{I}_S= \{I \subseteq S: I \in \mathcal{I}\}$).

	    Differently from the previous two sections 
	    where we have applied the greedy algorithm 
	    based on 
	    the reverse edge arrival order (Algorithm~\ref{greedy_part}), here we introduce a new idea 
	    to build a feasible $b$-matching. Observe that 
	    after the streaming phase, $Top(Q_{\mathcal{M}})$ is 
	    an independent set (more precisely, a base of $\mathcal{M}|S$). 
	    Let $M=Top(Q_{\mathcal{M}})$. If $M$ is a $b$-matching in 
	    the hypergraph, we are done. If not, some vertex $v$ must have more than $b_v$ incident edges in $M$. We discard one such edge $e$ in $M$ (the choice of $e$ will depend 
	    on the edge arrival order, see below) and 
	    replace it with another edge $e'$ that arrived earlier than $e$ (sometimes $e$ is not replaced at all, as we will explain below). We guarantee $M \cup \{e'\} \backslash \{e\}$ 
	    remains independent in $\mathcal{M}$. The same procedure 
	    is repeated iteratively, where we replace newer elements of the independent set $M$ by the older elements, 
	    until $M$ becomes a feasible $b$-matching in the hypergraph. A formal description of this procedure is provided in Algorithm~\ref{greedy_part_matroid_submod}.
	    
	    We now give a more detailed description. 
	    In Algorithm~\ref{greedy_part_matroid_submod}, an element $e$ is \emph{dominated} in the hypergraph by an element $e_d \in M$ if there exists $v \in V$ and $1 \leq i \leq b_v$ such that $e$ and $e_d$ are both in $Q_{v, i}$ and $e_d$ arrived later than $e$ in $S$. The algorithm starts with the maximum weight independent set $M = Top(Q_{\mathcal{M}})$ in $\mathcal{M}|S$, such that $w_{\mathcal{M}}(M) = g(S)$. Then, at each step, if the independent family $M$ is not a $b$-matching, then some edge in $M$ must be dominated by some other edge in $M$. We take among these dominated elements the latest one that arrived in $S$, and we replace it by an element $e'$ of $C_e \backslash \{e\}$ that is not spanned by $M \backslash \{e\}$ (where $C_e$ is defined as the circuit that is created by $e$ in $Top(Q_{\mathcal{M}})$ when it is inserted, see Line~11 in Algorithm~\ref{streaming_part_matroid_submod}; if $C_e = \emptyset$ then we remove $e$ without replacing it). Such an element $e'$ always exists, otherwise $e$ would have been spanned by $M \backslash \{e\}$ (see Lemma~\ref{lem:e-well-defined}). Moreover this element satisfies $w_{\mathcal{M}}(e') \geq w_{\mathcal{M}}^*(e)$ (because $w_{\mathcal{M}}^*(e)$ is the minimum reduced weight among the reduced weights of the elements in $C_e \backslash \{e\}$) and we know that $g(e)$ is already ``paid for'' by the element $e_d$ which dominates $e$ (if $C_e = \emptyset$ then $w_{\mathcal{M}}^*(e) = 0$ and there is no need to add any element $e'$ to compensate the loss). The element $e'$ that replaces $e$ arrived earlier than $e$ in $S$ because it is part of the circuit that was created when $e$ arrived, so this newly arrived element $e'$ cannot dominate $e_d$. As $e$ was the latest dominated element, all future added elements will be elements that arrived earlier than $e$ in $S$, so $e_d$ will never be dominated in the next steps and $g(e)$ will be ``paid for'' by $e_d$ until the end of the algorithm.
	
	    Regarding the memory consumption of the algorithm, it is easy to see that the number of variables used will be $O(\log_{1 + \varepsilon}(W/\varepsilon) \cdot \min\{|M_{max}|, r_{\mathcal{M}}\})$, where $|M_{max}|$ denotes the maximum cardinality matching and $W$ denotes the maximum quotient $\frac{f(e \,|\, Y)}{f(e' \,|\, X)}$, for $X \subseteq Y \subseteq E$, $e, e' \in E$, $f(e' \,|\, X) > 0$ 
	    (the idea is similar to the one in Appendix~\ref{appendix_memory}). 
	    In fact, we do not need to actually store the circuits $C_e$, as we can get these circuits back during the construction phase by going back in time in the structure of the queues $Q_{\mathcal{M}}$.

		\begin{algorithm}
		\caption{Streaming phase for Matroid-constrained Maximum Submodular $b$-Matching}\label{streaming_part_matroid_submod}
		\begin{algorithmic}[1]
		\State $S \gets \emptyset$
		\State $Q_{\mathcal{M}} \gets (Q_{\mathcal{M},1} = \emptyset, \cdots, Q_{\mathcal{M}, r_{\mathcal{M}}} = \emptyset)$
		\State $\forall v \in V : Q_v \gets (Q_{v,1} = \emptyset, \cdots, Q_{v, b_v} = \emptyset)$ 
		\For{$e = e_t,\, 1 \leq t \leq |E|$ an edge from the stream}
			\For{$u \in e$}
				\State $w_u^*(e) \gets \min \{w_u(Q_{u, q}.top()) : 1 \leq q \leq b_u\}$
				\State $q_u(e) \gets q \text{ such that $w_u(Q_{u, q}.top()) = w_u^*(e)$}$
			\EndFor
			
			\If{$Top(Q_{\mathcal{M}})\cup \{e\} \in \mathcal{I}$, \emph{i.e.}  $C_e = \emptyset$}
			   \State $w_{\mathcal{M}}^*(e) \gets 0$
			   \State $q_{\mathcal{M}}(e) \gets q \text{ such that $Q_{\mathcal{M}, q}$ is empty}$
			\EndIf
			\If{$Top(Q_{\mathcal{M}})\cup \{e\}$ contains a circuit $C_e$}
			  \State $w_{\mathcal{M}}^*(e) \gets \min_{e' \in C\backslash \{e\}}w_{\mathcal{M}}(e')$
			  \State $q_{\mathcal{M}}(e) \gets q$ such that $w_{\mathcal{M}}(Q_{\mathcal{M}, q}.top())$ is equal to $\min_{e' \in C\backslash \{e\}}w_{\mathcal{M}}(e')$ and $Q_{\mathcal{M},q}.top() \in C$
			\EndIf   
		
			\If {$f(e \,|\, S) > \alpha(\sum_{u \in e} w_u^*(e) + w_{\mathcal{M}}^*(e))$}
			    \State $\pi \gets \text{a random variable equal to $1$ with probability $p$ and $0$ otherwise}$
				\If{$\pi = 0$} \textbf{continue} \Comment skip edge $e$ with probability $1 - p$
				\EndIf
				\State $g(e) \gets f(e \,|\, S) - \sum_{u \in e} w_u^*(e) - w_{\mathcal{M}}^*(e)$
				\State $S \gets S \cup \{e\}$
				\For{$u \in e$}
					\State $w_u(e) \gets w_u^*(e) + g(e)$
					\State $r_u(e) \gets Q_{u, q_u(e)}.top()$
					\State $Q_{u, q_u(e)}.push(e)$
				\EndFor
				\State $w_{\mathcal{M}}(e) \gets w_{\mathcal{M}}^*(e) + g(e)$
				\State $r_{\mathcal{M}}(e) \gets Q_{{\mathcal{M}}, q_{\mathcal{M}}(e)}.top()$
				\State $Q_{{\mathcal{M}}, q_{\mathcal{M}}(e)}.push(e)$
			\EndIf
		\EndFor
		\end{algorithmic}
		\end{algorithm}
		
		\begin{algorithm}
		\caption{Construction phase for Matroid-constrained Maximum Submodular $b$-Matching}\label{greedy_part_matroid_submod}
		\begin{algorithmic}[1]
		\State $M \gets Top(Q_{\mathcal{M}})$
		\While{some element in $M$ is dominated in the hypergraph by another element in $M$}
		    \State let $e$ be the latest dominated element in $M$
		    \If{$C_e = \emptyset$}
		    \State $M \gets M \backslash \{e\}$
		    \Else
		    \State let $e'$ be an element in $C_e \backslash (\text{span}(M \backslash \{e\}) \cup \{e\})$ 
		    \State $M \gets (M \cup \{e'\})\backslash \{e\}$
		    \EndIf
		\EndWhile
		\Return $M$
		\end{algorithmic}
		\end{algorithm}
		
	\subsection{Analysis for Linear Function Maximization}
	    In this section, $p = 1$. For each discarded elements $e \in E \backslash S$, similarly as before, we set $w_{\mathcal{M}}(e) = w_{\mathcal{M}}^*(e)$. Moreover, for $e = e_t$ we set  $w(e) = f(e\,|\,\emptyset) = f(e\,|\, S^{(t)})$ the weight of an element.
	    
	    We introduce some basic facts in matroid theory, e.g., see \cite{Sch2003}. 
	    
	    \begin{proposition} 
	    \label{pro:matroid}
	        Given a matroid $\mathcal{M}=(E,\mathcal{I})$ with weight $w:E \rightarrow \mathbb{R_+}$, then 
	    
            \begin{enumerate}[(i)]
    	        \item An independent set $I \in \mathcal{I}$ is a maximum
    	        weight base if and only if, for every element 
    	        $e \in E \backslash I$, $I\cup \{e\}$ contains a circuit 
    	        and $w(e) \leq \min_{e' \in C\backslash \{e\}}w(e')$.
    	        \item If $I \in \mathcal{I}$, $I\cup \{e\}$ contains a circuit 
    	        $C_1$ and $I\cup \{e'\}$ contains a circuit 
    	        $C_2$ and $C_1$ and $C_2$ contain a common element $e'' \in I$, then there exists another circuit 
    	        $C_3 \subseteq (C_1 \cup C_2)\backslash \{e''\}.$
    	        \item If $X, Y \subseteq E$, $z \in E$ are such that $X \subseteq \text{span}(Y)$ and $z \in \text{span}(X)$, then $z \in \text{span}(Y)$.
    	    \end{enumerate}
	    
	    \end{proposition}
	    
		\begin{lemma} \label{lem:maximum_weight_base}	
		    Let $\{e_1,\cdots, e_t\}$ be the set of edges arrived so far. Then the top elements $Top(Q_{\mathcal{M}}) = Top(Q_{\mathcal{M}}^{(t)})$ forms a maximum weight base in $\{e_1,\cdots, e_t\}$ with regard to the reduced weight $w_{\mathcal{M}}$.
		\end{lemma}
		
		\begin{proof}
		    This can be easily proved by induction on the number of edges arrived so far and Proposition~\ref{pro:matroid}(i) and (ii).
		\end{proof}
		
		\begin{corollary} \label{cor:maximum_weight_base}
			 At the end of the algorithm, $w_{\mathcal{M}}(Q_{\mathcal{M}}) \geq w_{\mathcal{M}}(M^{opt})$.
		\end{corollary}

	    \begin{lemma}
	        It holds that $(1 + \varepsilon)(k + 1) g(S) \geq w(M^{opt})$.
	    \end{lemma}
	    
	    \begin{proof}
	        Using the fact that $w(e) \leq (1 + \varepsilon)\left(\sum_{u \in e}w_u(e) + w_{\mathcal{M}}(e)\right)$ (observe that it is true when $e$ is stored in $S$ as well as when it is not) we get:
	        \begin{align*}
	            w(M^{opt}) &\leq (1 + \varepsilon) \sum_{e \in M^{opt}} \left(\sum_{u \in e}w_u(e) + w_{\mathcal{M}}(e)\right)\\
	            &= (1 + \varepsilon) \sum_{v \in V} w_v(M^{opt} \cap \delta(v)) + (1 + \varepsilon) w_{\mathcal{M}}(M^{opt})\\
	            &\leq (1 + \varepsilon) \sum_{v \in V} g(S \cap \delta(v)) + (1 + \varepsilon) w_{\mathcal{M}}(Q_{\mathcal{M}})\\ 
	            &= (1 + \varepsilon) (kg(S) + g(S)).
	        \end{align*}
	        For the last inequality we use Corollary~\ref{cor:maximum_weight_base} and Proposition~\ref{lemma_g_w}.
	    \end{proof}
	    
	    Now we turn our attention to Algorithm~\ref{greedy_part_matroid_submod}. First, we check that this algorithm performs legitimate operations and terminates.
	    
	    \begin{lemma} \label{lem:e-well-defined}
	        \begin{enumerate}[(i)]
	            \item The element $e'$ defined at Line~7 in Algorithm~\ref{greedy_part_matroid_submod} is well-defined.
	            \item Through the execution of the algorithm, $M$ remains an independent set in $\mathcal{M}$.
	        \end{enumerate}
	    \end{lemma}
	    
	    \begin{proof}
	        To prove (i), suppose $e'$ does not exist. If such an element does not exist, then it means that we have the inclusion $C_e \backslash \{e\} \subseteq \text{span}(M \backslash \{e\})$. In addition, $C_e$ is a circuit so $e \in \text{span}(C_e \backslash \{e\})$. As a result, we should have $e \in \text{span}(M \backslash \{e\})$ (see Proposition~\ref{pro:matroid}(iii)), in contradiction with the fact that $M$ is independent in $\mathcal{M}$. 
	        
	        For the point (ii), observe that, at the beginning, $M$ is an independent set because of Lemma~\ref{lem:maximum_weight_base}. Then, if $M$ is an independent set, then so is $M\backslash\{e\}$. Moreover, if $e$ is replaced by an element $e' \in C_e \backslash (\text{span}(M \backslash \{e\}) \cup \{e\})$, the set $(M \backslash \{e\}) \cup \{e'\}$ is still an independent set, as $e'$ is not in the span of $M \backslash \{e\}$.
	    \end{proof}
	    
	    \begin{lemma}
	        Algorithm~\ref{greedy_part_matroid_submod} terminates and outputs a feasible $b$-matching.
	    \end{lemma}
	    
	    \begin{proof}
	        The algorithm terminates because an element is either not not replaced or replaced by an element that arrived earlier in $S$. By Lemma~\ref{lem:e-well-defined}(ii), through the execution of the algorithm, $M$ remains an independent set in the matroid. In addition, when there is no domination relation between elements in $M$, then $M$ must be a $b$-matching (as there is at most one element per queue) in the hypergraph.
	    \end{proof}
	    
	    Let us define the set $D$ of elements that were removed so far from $M$ during the execution of Algorithm~\ref{greedy_part_matroid_submod} (because they were dominated in the hypergraph).
	    
	    \begin{lemma}
	        At any time during the execution of Algorithm~\ref{greedy_part_matroid_submod}, the elements in $D$ are dominated by some elements in $M$.
	    \end{lemma}
	    
	    \begin{proof}
	        When an element $e$ is added to $D$, it means that at that point it is the latest dominated element in $M$, and that it is dominated by some other element $e_d$ that arrived later. As elements are always replaced by other elements that arrived earlier, and because all the subsequent dominated elements considered by the algorithm will be edges that arrived earlier than $e$, this edge $e_d$ will not be removed from $M$ until the end of the execution the algorithm.
	    \end{proof}
	    
	    \begin{lemma}
	        At any time during the execution of Algorithm~\ref{greedy_part_matroid_submod}, it holds that $w(M) \geq w_{\mathcal{M}}(M) + g(D)$.
	    \end{lemma}
	    
	    \begin{proof}
	        We have
	        \begin{align*}
	            w(M) &= w_{\mathcal{M}}(M) + \sum_{e \in M} \sum_{v \in e} g(\{\text{$e'$ dominated by $e$ in $Q_{v, q_v(e)}$}\})\\
	            &\geq w_{\mathcal{M}}(M) + g(D),
	        \end{align*}
	        because, by the previous lemma, $D$ is a subset of the elements that are actually dominated in the hypergraph (and an element can be dominated multiple times in the real computation of $w(M)$).
	    \end{proof}
	    
	    \begin{lemma}
	        At any time during the execution of Algorithm~\ref{greedy_part_matroid_submod}, $w_{\mathcal{M}}(M) + g(D) \geq g(S)$.
	    \end{lemma}
	    
	    \begin{proof}
	        We prove by induction on the number of times the while loop 
	        is executed. At the beginning, $M = Top(Q_{\mathcal{M}})$ so $w(M) \geq w_{\mathcal{M}}(Q_{\mathcal{M}}) = g(S)$, and $D = \emptyset$. For the induction step,  suppose that we have to remove $e \in M$ from $M$ because it is dominated by $e_d \in M$.
	        %in some queue $Q_{v, i}$. 
	        %It means that $g(e)$ is ``paid for'' in the sum of the elements %below $e_d$ in $Q_{v, i}$, which is $w_v^*(e_d)$, and therefore this cost is transferred to $g(D)$. 
	        If $C_e = \emptyset$ it means that $w_{\mathcal{M}}(e) = g(e)$, 
	        then $w_{\mathcal{M}}((M \backslash \{e\})) + g(D \cup \{e\})
	        = w_{\mathcal{M}}(M) + g(D) \geq g(S)$. If $C_e \neq \emptyset$, 
	        then the replacing element $e'$
	        guarantees that $(M \backslash \{e\}) \cup \{e'\}\in \mathcal{I}$ and $w_{\mathcal{M}}(e') \geq w_{\mathcal{M}}^*(e)$ so the new value $w_{\mathcal{M}}((M \backslash \{e\}) \cup \{e'\}) + g(D \cup \{e\}) = w_{\mathcal{M}}(M) - g(e) - w_{\mathcal{M}}^*(e) + w_{\mathcal{M}}(e') + g(D) + g(e) \geq w_{\mathcal{M}}(M) + g(D) \geq g(S)$ also respects the inequality. 
	    \end{proof}
	    
	    \begin{lemma}\label{lem:greedy_matroid}
	        Algorithm~\ref{greedy_part_matroid_submod} returns a feasible $b$-matching $M$ such that $w(M) \geq g(S)$. 
	    \end{lemma}
	    
	    \begin{proof}
	        This follows easily from the previous lemmas.
	    \end{proof}
	    
	    Therefore combining the previous results we get a $(1 + \varepsilon)(k + 1)$ approximation streaming algorithm for the matroid-constrained maximum weight $b$-matching problem: 
	    \begin{theorem}
	        Algorithm~\ref{streaming_part_matroid_submod} combined with Algorithm~\ref{greedy_part_matroid_submod} provides a $(1 + \varepsilon)(k + 1)$ approximation for the maximum weight $b$-matching problem under a matroid constraint.
	    \end{theorem}
	    
	    \begin{remark}
	        When the hypergraph is $k$-partite, our algorithm builds the set $S$ exactly as the algorithm by Garg et al.~\cite{GJS2021} does. As a result, it proves the conjecture they formulated in the case where we have $k$ partition matroids and one general matroid: we can extract from $S$ a $k+1$ approximation in polynomial time, if $\varepsilon$ is set to $0$.
	    \end{remark}
		
	\subsection{Analysis for Monotone Submodular Function Maximization}
	    In this section, $p = 1$ and $\varepsilon > 0$. The next two lemmas relate the total gain $g(S)$ with 
		$f(S\,|\,\emptyset)$ and $f(M^{opt} \,|\, S)$.
		
		\begin{lemma} \label{lemma_gSp_fSp}
		    It holds that $g(S) \geq \frac{\varepsilon}{1 + \varepsilon}f(S\,|\,\emptyset)$.
		\end{lemma}
		
		\begin{proof}
		    Same proof as for Lemma~\ref{lemma_g_f_ineg}.
		\end{proof}

		\begin{lemma} \label{lem:another}
		    It holds that $(1 + \varepsilon)(k + 1)g(S) \geq f(M^{opt} \,|\, S)$.
		\end{lemma}
		
		\begin{proof}
		    By the same argument as in the proof of Lemma~\ref{lemma_2g_geq_f_MS}, we have 
		    \[\sum_{e \in M^{opt} \backslash S} (1 + \varepsilon) \sum_{u \in e} w_u^*(e) \leq (1 + \varepsilon)k \cdot g(S).\]
		    Moreover, Corollary~\ref{cor:maximum_weight_base} shows that $g(S) = w_{\mathcal{M}}(Q_{\mathcal{M}}) \geq w_{\mathcal{M}}(M^{opt})$ and we know that
		    \[w_{\mathcal{M}}(M^{opt}) \geq \sum_{e \in M^{opt}\backslash S} w_{\mathcal{M}}(e) = \sum_{e \in M^{opt}\backslash S} w_{\mathcal{M}}^*(e).\]
		    As a result, we obtain
		    \begin{align*}
		        f(M^{opt} \,|\, S) &\leq \sum_{e \in M^{opt}\backslash S} f(e \,|\, S) 
		        \leq \sum_{e = e_t \in M^{opt}\backslash S} f(e \,|\, S^{(t-1)})\\
		        &\leq (1+\varepsilon)\left(\sum_{e \in M^{opt}\backslash S} \sum_{u \in e} w_u^*(e) +   w_{\mathcal{M}}^*(e)\right)\\
		        &\leq (1 + \varepsilon)(k + 1)g(S).
		    \end{align*}
		\end{proof}

		\begin{lemma} It holds that 
		    $\left((1+ \varepsilon) (k + 1) + 1 + \frac{1}{\varepsilon}\right)g(S) \geq f(M^{opt} \,|\, \emptyset).$
		\end{lemma}
		
		\begin{proof}
		    By Lemmas~\ref{lemma_gSp_fSp} and~\ref{lem:another}, we get $\left((1+ \varepsilon) (k + 1) + 1 + \frac{1}{\varepsilon}\right)g(S) \geq f(M^{opt} \,|\, S) + f(S \,|\, \emptyset) = f(M^{opt} \cup S \,|\, \emptyset) \geq f(M^{opt} \,|\, \emptyset)$ because $f$ is monotone.
		\end{proof}
		
		\begin{lemma}
		    With $S$ as input, Algorithm~\ref{greedy_part_matroid_submod} returns a feasible  $b$-matching $M$ satisfying $f(M) \geq g(S) + f(\emptyset)$.
		\end{lemma}
		
		\begin{proof}
		    Proof similar to the one for Lemma~\ref{lemma_greedy_submod} (using Lemma~\ref{lem:greedy_matroid} instead of Lemma~\ref{lemma_greedy_weight}).
		\end{proof}
		
		As a result, we get the following theorem (the same way we obtained Theorem~\ref{theo_submod_mono}):
		\begin{theorem}
		\label{theo_submo_mono_matroid}
		    For non-negative monotone submodular functions, Algorithm~\ref{streaming_part_matroid_submod} with $p = 1$ combined with Algorithm~\ref{greedy_part_matroid_submod} provides a feasible $b$-matching such that
		    \[\left((1+ \varepsilon) (k + 1) + 1 + \frac{1}{\varepsilon}\right)f(M) \geq f(M^{opt}).\]
	    	By setting $\varepsilon = \frac{1}{\sqrt{k + 1}}$, we attain the approximation ratio of $k + 2\sqrt{k+1} + 2$ for all $k$. 
		\end{theorem}
		
    \subsection{Analysis for Non-Monotone Submodular Function Maximization}
    
        In this section, we assume $\frac{1}{1 + (k + 1)(1 + \varepsilon)} \leq p \leq \frac{1}{k + 1}$. The following lemma is the counterpart of Lemma~\ref{lemma_eg_ef}, 
        whose proof again uses the technique of conditioning (in a more 
        general form). 
        
        \begin{lemma}
            It holds that $\left((1 + \varepsilon)(k + 1) + \frac{1 + \varepsilon}{\varepsilon}\right)\mathbb{E}[g(S)] \geq \mathbb{E}[f(S \cup M^{opt} \,|\, \emptyset)]$.
        \end{lemma}
        
        \begin{proof}
        
            By Lemma~\ref{lemma_gSp_fSp}, for any realization on randomness, we have  $\frac{1 + \varepsilon}{\varepsilon}g(S) \geq f(S\,|\,\emptyset)$, so the inequality also holds in expectation.
            
            We next show that, for any $e \in M^{opt}$ we have 
            \begin{equation}\mathbb{E}[f(e \,|\, S)] \leq (1 + \varepsilon)\mathbb{E}\left[\sum_{u \in e}w_u(e) + w_{\mathcal{M}}(e)\right]. 
            \label{equ:thirdOne}
            \end{equation}
            
            Let $e \in M^{opt}$. Conditioning on $A_e = [f(e \,|\, S^{(t-1)}) \leq (1 + \varepsilon)(\sum_{u \in e}w_u^*(e) + w_{\mathcal{M}}^*(e))]$ (\emph{i.e.} $e$ cannot be part of $S$), we have 
            \begin{align*}
				\mathbb{E}[f(e \,|\, S)\,|\,A_e] &\leq \mathbb{E}[f(e \,|\, S^{(t-1)})\,|\,A_e]\\
				&\leq \mathbb{E}\left[(1 + \varepsilon)(\sum_{u \in e}w_u^*(e) +  w_{\mathcal{M}}^*(e))\,|\,A_e\right]\\
				&= (1 + \varepsilon)\mathbb{E}\left[\sum_{u \in e}w_u(e) + w_{\mathcal{M}}(e)\,|\,A_e\right].
			\end{align*}
			
			For the condition $\overline{A_e}$, recall that it means 
			that with probability $p$, the edge is added into $S$ 
			and with probability $1-p$, it is not. Thereby we obtain
			\begin{align*}
			    &\mathbb{E}\left[\sum_{u \in e}w_u(e) + w_{\mathcal{M}}(e)\,|\,\overline{A_e}\right]\\
				=\,&p \cdot \mathbb{E}\left[(k + 1) f(e \,|\, S^{(t-1)}) - k \left(\sum_{u \in e}w_u^*(e) + w_{\mathcal{M}}^*(e)\right) \,|\,\overline{A_e}\right]\\
				&\,+ (1 - p)\cdot\mathbb{E}\left[\sum_{u \in e}w_u^*(e) +  w_{\mathcal{M}}^*(e)\,|\,\overline{A_e}\right]\\
				=\,& (k + 1) \cdot p \cdot \mathbb{E}\left[f(e \,|\, S^{(t-1)})\,|\,\overline{A_e}\right]\\
				&\,+ (1 - (k + 1) \cdot p)\cdot\mathbb{E}\left[\sum_{u \in e}w_u^*(e) + w_{\mathcal{M}}^*(e)\,|\,\overline{A_e}\right]\\
				\geq\,&(k + 1) \cdot p \cdot \mathbb{E}\left[f(e \,|\, S^{(t-1)})\,|\,\overline{A_e}\right],
			\end{align*}
			where in the last inequality we use the fact that $p \leq \frac{1}{k + 1}$.
			
			Now as $p \geq \frac{1}{1 + (k + 1)(1 + \varepsilon)}$, we have 
            \begin{align*}
                (1 + \varepsilon)\mathbb{E}\left[\sum_{u \in e}w_u(e) +  w_{\mathcal{M}}(e)\,|\,\overline{A_e}\right]
                &\geq (1 + \varepsilon)(k + 1) \cdot p \cdot \mathbb{E}\left[f(e \,|\, S^{(t-1)})\,|\,\overline{A_e}\right]\\
                &\geq(1 - p) \cdot \mathbb{E}\left[f(e \,|\, S^{(t-1)})\,|\,\overline{A_e}\right]\\
                &\geq\mathbb{E}\left[f(e \,|\, S)\,|\,\overline{A_e}\right],  
            \end{align*}
            and we have established~(\ref{equ:thirdOne}). 
            
            Similar to the proof of Lemma \ref{lemma_eg_ef} we derive
            \begin{align*}
            	\mathbb{E}\left[f(M^{opt} \,|\, S)\right] &\leq \sum_{e \in M^{opt}} \mathbb{E}\left[f(e \,|\, S)\right]\\
                &\leq (1 + \varepsilon) \sum_{e \in M^{opt}} \mathbb{E}\left[\sum_{u \in e} w_u(e)\right] + (1 + \varepsilon) \mathbb{E}[w_{\mathcal{M}}(M^{opt})]\\	
                &\leq k(1 + \varepsilon)\mathbb{E}[g(S)] + (1 + \varepsilon) \mathbb{E}[w_{\mathcal{M}}(M^{opt})].
            \end{align*}
            By Lemma~\ref{cor:maximum_weight_base} we know that, for any realization of randomness, we have $w_{\mathcal{M}}(M^{opt}) \leq w_{\mathcal{M}}(S_f) = g(S)$. Thus we get $\mathbb{E}\left[f(M^{opt} \,|\, S)\right] \leq (k + 1)(1 + \varepsilon)\mathbb{E}[g(S)]$. 
            
            Now the bound on $\mathbb{E}\left[f(M^{opt} \,|\, S)\right]$ 
            and the bound on $\mathbb{E}\left[f(S \,|\, \emptyset)\right]$ (argued in the beginning of the proof) give us the lemma. 
        \end{proof}
        
        Finally, using Lemma~\ref{lemma_ineg_h}, we obtain the following theorem.
        \begin{theorem}
        \label{theo_submo_nonmono_matroid}
            For non-negative submodular functions, Algorithm~\ref{streaming_part_matroid_submod} executed with $p = \frac{1}{1 + (k + 1)(1 + \varepsilon)}$ combined with Algorithm~\ref{greedy_part} provides a $b$-matching $M$ independent in $\mathcal{M}$ such that:
            \[\frac{1 + (k + 1)(1 + \varepsilon)}{(k + 1)(1 + \varepsilon)}\left((1+ \varepsilon) (k + 1) + 1 + \frac{1}{\varepsilon}\right)\mathbb{E}[f(M)] \geq f(M^{opt}).\]
            This ratio is optimized when $\varepsilon = \frac{\sqrt{k + 2}}{k + 1}$, providing a $k +2\sqrt{k+2} + 3$ approximation.
        \end{theorem}
        
\section{Matchoid-constrained Maximum Submodular $b$-Matching}
\label{sec:matchoid}
    In this section we consider an even more general case of a $b$-matching on a $k$-uniform hypergraph on which we impose a $k'$-matchoid constraint $\mathcal{M} = (\mathcal{M}_i=(E_i,\mathcal{I}_i))_{i=1}^s$. A matching $M \subseteq E$ is feasible only if it respects the capacities of the vertices and $E_i \cap M$ is an independent set in $\mathcal{M}_i$ for all $i$. 
		
	\subsection{Description of the Algorithm}
	
	    For the streaming phase, our algorithm, formally described in Algorithm~\ref{streaming_part_matchoid_submod}, is an adaptation of Algorithm~\ref{streaming_part_submod}. 
	    We let $\alpha = 1 + \varepsilon$ and $\gamma > 1$. Here $S$ 
    	as before denotes the set of elements stored so far. Additionally, 
    	we maintain a set $S_f \subseteq S$, an independent set in 
    	$\mathcal{M}$. The set $S_f$ is, in a sense, equivalent 
    	to the set of top elements $Top(Q_{\mathcal{M}})$ in the previous 
    	section when $\mathcal{M}$ was a simple matroid. 
    	Note that here we do not use any special data structures 
    	associated with the matchoid $\mathcal{M}$. 
    	For a newly-arrived element $e$, its reduced weight 
    	in the matchoid $\mathcal{M}$ is determined by the 
    	set of circuits contained in $S_f \cup \{e\}$ (see Lines~8-11). If the new element $e$ is to be stored 
    	(see Lines 8-11, 17 and 23), it replaces 
    	the lightest elements of the circuits in $S_f \cup \{e\}$.\footnote{We note that how the elements in $S_f$ 
    	are updated here are based on the same idea as in 
    	~\cite{DBLP:journals/mp/ChakrabartiK15,ChekuriGQ15} 
    	when constraint is only a $p$-matchoid without the hypergraph.}
    	In the end, the set $S_f$ (instead of all stored elements $S$) is given to the simple greedy 
    	Algorithm~\ref{greedy_part} as in Sections~\ref{sec:weight} and~\ref{sec:submodular} to produce a feasible $b$-matching.

	    Here we give some intuition. We retain only the elements in $S_f$ (as input to the greedy algorithm) because they are independent in every matroid $\mathcal{M}_i$ (hence any subset of them chosen by the greedy algorithm), releasing us from the worry that the output is not independent in $\mathcal{M}$. 
	    We use $\gamma>1$ to decide the relative importance of the matchoid $\mathcal{M}$ (see Line 13). 
	    A larger $\gamma$ guarantees that the gains of new edges grows quickly when $S_f$ is updated. 
	    By doing this, $S_f$, by itself, contributes to a significant fraction of all gains in $g(S)$ (see Lemma~\ref{lemma_gS_gSp-matchoid}). However, an overly large $\gamma$ causes us to throw away too many edges, thus hurting the final approximation ratio. To optimize, we need to choose $\gamma$ carefully.
	
	    Regarding the memory consumption of the algorithm, it is easy to see that the number of variables used would be a $O(\log_{1 + \varepsilon}(W/\varepsilon) \cdot |M_{max}|)$, where $M_{max}$ denotes the maximum cardinality $b$-matching in the graph and $W$ denotes the maximum quotient $\frac{f(e \,|\, Y)}{f(e' \,|\, X)}$, for $X \subseteq Y \subseteq E$, $e, e' \in E$, $f(e' \,|\, X) > 0$.
		
		\begin{algorithm}[h]
		\caption{Streaming phase for Matchoid-constrained Maximum Submodular $b$-Matching}\label{streaming_part_matchoid_submod}
		\begin{algorithmic}[1]
		\State $S \gets \emptyset$
		\State $S_f \gets \emptyset$
		\State $\forall v \in V : Q_v \gets (Q_{v,1} = \emptyset, \cdots, Q_{v, b_v} = \emptyset)$ 
		\For{$e = e_t,\, 1 \leq t \leq |E|$ an edge from the stream}
			\For{$u \in e$}
				\State $w_u^*(e) \gets \min \{w_u(Q_{u, q}.top()) : 1 \leq q \leq b_u\}$
				\State $q_u(e) \gets q \text{ such that $w_u(Q_{u, q}.top()) = w_u^*(e)$}$
			\EndFor
			
			\State $C_e \gets \emptyset$
			\For{$i$ such that $e \in E_i$}
    			\If{$(S_f \cap E_i)\cup \{e\}$ contains a circuit $X$ in $\mathcal{M}_i$}
    			  \State $C_e \gets C_e \cup \{\text{arg}\min_{e' \in X\backslash \{e\}}w_{\mathcal{M}}(e')\}$
    			\EndIf
    		\EndFor
    		\State $w_{\mathcal{M}}^*(e) \gets w_{\mathcal{M}}(C_e)$
		
			\If {$f(e \,|\, S) > \alpha(\sum_{u \in e} w_u^*(e) + \gamma \cdot w_{\mathcal{M}}^*(e))$}
			    \State $\pi \gets \text{a random variable equal to $1$ with probability $p$ and $0$ otherwise}$
				\If{$\pi = 0$} \textbf{continue} \Comment skip edge $e$ with probability $1 - p$
				\EndIf
				\State $g(e) \gets f(e \,|\, S) - \sum_{u \in e} w_u^*(e) - w_{\mathcal{M}}^*(e)$
				\State $S \gets S \cup \{e\}$
				\For{$u \in e$}
					\State $w_u(e) \gets w_u^*(e) + g(e)$
					\State $r_u(e) \gets Q_{u, q_u(e)}.top()$
					\State $Q_{u, q_u(e)}.push(e)$
				\EndFor
				\State $w_{\mathcal{M}}(e) \gets w_{\mathcal{M}}^*(e) + g(e)$
				\State $S_f \gets (S_f \cup \{e\}) \backslash C_e$
			\EndIf
		\EndFor
		\State update the values of $r_v$ for $v \in V$ as necessary so that only the elements in $S_f$ are considered 
		\end{algorithmic}
		\end{algorithm}
		
	\subsection{Analysis for Monotone Submodular Function Maximization}
	    In this section, $p = 1$. For each discarded element
	    $e \in E \backslash S$, similarly as before, we set $w_{\mathcal{M}}(e) = w_{\mathcal{M}}^*(e)$. Moreover, we define an auxiliary reduced weight $w_{\mathcal{M}}'$ such that $w_{\mathcal{M}}'(e) = w_{\mathcal{M}}(e)$ if $e \in S$ and $w_{\mathcal{M}}'(e) = \gamma \cdot (1 + \varepsilon) \cdot w_{\mathcal{M}}(e)$ otherwise.
	    
	    The way $S_f$ is built can be seen as an execution of the algorithm of Chekuri et al.~\cite{ChekuriGQ15} where elements of $E$ are given weights $w_{\mathcal{M}}'$. Using their results we know that:
	    \begin{proposition}
	    \label{pro:matchoid}
	    Let $M' \subseteq E$ satisfying $M' \cap E_i \in \mathcal{I}_i$ for all $i$. Then we have
	    \[\frac{(\gamma \cdot (1+\varepsilon))^2}{\gamma \cdot (1+\varepsilon)-1} k' \cdot w_{\mathcal{M}}(S_f) = \frac{(\gamma \cdot (1+\varepsilon))^2}{\gamma \cdot (1+\varepsilon)-1} k' \cdot w_{\mathcal{M}}'(S_f) \geq w_{\mathcal{M}}'(M')\]
	    and (this inequality is from Lemma 11 of the arXiv version of~\cite{ChekuriGQ15})
	    \[\left(\frac{(\gamma \cdot (1+\varepsilon))^2}{\gamma \cdot (1+\varepsilon)-1} k' - \frac{\gamma \cdot (1+\varepsilon)}{\gamma \cdot (1+\varepsilon)-1}\right) \cdot w_{\mathcal{M}}(S_f) \geq \gamma \cdot (1 + \varepsilon) \cdot w_{\mathcal{M}}(M'\backslash S).\]
	    \end{proposition}
		
		The next two lemmas relate the total gain $g(S)$ with
		$f(S\,|\,\emptyset)$ and $f(M^{opt} \,|\, S)$.
		
		\begin{lemma} \label{lemma_gSp_fSp-matchoid}
		    It holds that
		    $g(S) \geq \frac{\varepsilon}{1 + \varepsilon}f(S\,|\,\emptyset)$.
		\end{lemma}
		
		\begin{proof}
		    Same proof as for Lemma~\ref{lemma_g_f_ineg}.
		\end{proof}

		\begin{lemma} \label{lem:another-matchoid}
		    It holds that \[\left((1 + \varepsilon)k + \frac{(\gamma \cdot (1+\varepsilon))^2}{\gamma \cdot (1+\varepsilon)-1} k' - \frac{\gamma \cdot (1+\varepsilon)}{\gamma \cdot (1+\varepsilon)-1}\right)g(S) \geq f(M^{opt} \,|\, S).\]
		\end{lemma}
		
		\begin{proof}
		    By the same argument as in the proof of Lemma~\ref{lemma_2g_geq_f_MS}, we have 
		    \[\sum_{e \in M^{opt} \backslash S} (1 + \varepsilon) \sum_{u \in e} w_u^*(e) \leq (1 + \varepsilon)k \cdot g(S).\]
		    Moreover, we know that $g(S) = w_{\mathcal{M}}(S_f)$ and by Proposition~\ref{pro:matchoid}:
		    \[\left(\frac{(\gamma \cdot (1+\varepsilon))^2}{\gamma \cdot (1+\varepsilon)-1} k' - \frac{\gamma \cdot (1+\varepsilon)}{\gamma \cdot (1+\varepsilon)-1}\right) w_{\mathcal{M}}(S_f) \geq \gamma \cdot (1 + \varepsilon) \cdot w_{\mathcal{M}}^*(M^{opt} \backslash S).\]
		    As a result, we obtain
		    \begin{align*}
		        f(M^{opt} \,|\, S) &\leq \sum_{e \in M^{opt}\backslash S} f(e \,|\, S) 
		        \leq \sum_{e = e_t \in M^{opt}\backslash S} f(e \,|\, S^{(t-1)})\\
		        &\leq (1+\varepsilon)\left(\sum_{e \in M^{opt}\backslash S} \sum_{u \in e} w_u^*(e) +  \gamma \cdot w_{\mathcal{M}}^*(e)\right)\\
		        &\leq \left((1 + \varepsilon)k + \frac{(\gamma \cdot (1+\varepsilon))^2}{\gamma \cdot (1+\varepsilon)-1} k' - \frac{\gamma \cdot (1+\varepsilon)}{\gamma \cdot (1+\varepsilon)-1}\right)g(S).
		    \end{align*}
		\end{proof}
		
		The following lemma states that $S_f$ retains a reasonably 
		large fraction of the gains compared to $S$. For this lemma, we use the notion of \emph{erasing tree}. The erasing tree $T_e$ of an element $e \in S$ is defined recursively: it is made of the elements $e'$ of $C_e$ that were evicted from $S_f$ when $e$ was inserted and of the elements of the erasing trees $T_{e'}$ associated with these evicted elements.
		
		\begin{lemma} \label{lemma_gS_gSp-matchoid}
			It holds that $\left(1 + \frac{1}{\gamma \cdot (1 + \varepsilon) - 1}\right)g(S_f) \geq g(S)$.
		\end{lemma}
		
		\begin{proof}
			We have, for all element $e \in S_f$,
			\[g(e) \geq (\gamma \cdot (1 + \varepsilon) - 1) w_{\mathcal{M}}^*(e) = (\gamma \cdot (1 + \varepsilon) - 1)\sum_{e' \in T_e}g(e').\]
			where $T_e$ denotes the elements that are in the erasing tree of $e$, then
			\[g(S) = \sum_{e \in S_f} \left(g(e) + \sum_{e' \in T_e}g(e')\right) \leq \sum_{e \in S_f} \left(1 + \frac{1}{\gamma \cdot (1 + \varepsilon) - 1}\right)g(e),\]
			and the proof follows.
		\end{proof}
		
		\begin{lemma} It holds that
		    \begin{multline*}
		        \left(1 + \frac{1}{\gamma \cdot (1 + \varepsilon) - 1}\right) \cdot \left((1 + \varepsilon)k + \frac{(\gamma \cdot (1+\varepsilon))^2}{\gamma \cdot (1+\varepsilon)-1} k' - \frac{\gamma \cdot (1+\varepsilon)}{\gamma \cdot (1+\varepsilon)-1} + \frac{1 + \varepsilon}{\varepsilon}\right)g(S_f) \\\geq f(M^{opt} \,|\, \emptyset).
		    \end{multline*}
		\end{lemma}
		
		\begin{proof}
		    By Lemmas~\ref{lemma_gSp_fSp-matchoid} and~\ref{lem:another-matchoid}, we have that 
		    \begin{multline*}
		        \left(1 + \frac{1}{\gamma \cdot (1 + \varepsilon) - 1}\right) \cdot\left((1 + \varepsilon)k + \frac{(\gamma \cdot (1+\varepsilon))^2}{\gamma \cdot (1+\varepsilon)-1} k' - \frac{\gamma \cdot (1+\varepsilon)}{\gamma \cdot (1+\varepsilon)-1} + 1 + \frac{1}{\varepsilon}\right)g(S_f) \\
		        \geq f(M^{opt} \,|\, S) + f(S \,|\, \emptyset) = f(M^{opt} \cup S \,|\, \emptyset) \geq f(M^{opt} \,|\, \emptyset)
		    \end{multline*}
		    because $f$ is monotone. Then we use Lemma~\ref{lemma_gS_gSp-matchoid}.
		\end{proof}
		
		\begin{lemma}
		With $S_f$ as input, Algorithm~\ref{greedy_part} returns a feasible  $b$-matching $M$ with $f(M) \geq g(S_f) + f(\emptyset)$.
		\end{lemma}
		
		\begin{proof}
		    As argued in Lemma~\ref{lemma_greedy_weight}, $M$ respects the capacities. Furthermore, as $S_f$ is by construction an independent set in $\mathcal{M}$ and $M \subseteq S_f \in \mathcal{I}_i$ for al $i$, we have $M \in \mathcal{I}_i$ for all $i$. So $M$ is a feasible $b$-matching.
		    Finally, using an analysis similar to the one in the proof of  Lemma~\ref{lemma_greedy_submod}, we have $f(M) \geq g(S_f) + f(\emptyset)$ (the only difference being that now the "weight" $f(e_t \,|\, S^{(t-1)})$ of an element can be higher than the sum of the gains of the elements below it in the queues, which is not an issue for the analysis).
		\end{proof}
		
		As a result, we get the following theorem (the same way we obtained Theorem~\ref{theo_submod_mono}):
		\begin{theorem}
		\label{theo_submo_mono_matchoid}
		    For non-negative monotone submodular functions, Algorithm~\ref{streaming_part_matchoid_submod} with $p = 1$ combined with Algorithm~\ref{greedy_part} provides a feasible $b$-matching such that
		    \begin{multline*}
		        \left(1 + \frac{1}{\gamma \cdot (1 + \varepsilon) - 1}\right) \cdot \left((1 + \varepsilon)k + \frac{(\gamma \cdot (1+\varepsilon))^2}{\gamma \cdot (1+\varepsilon)-1} k' - \frac{\gamma \cdot (1+\varepsilon)}{\gamma \cdot (1+\varepsilon)-1} + 1 + \frac{1}{\varepsilon}\right)f(M) \\\geq f(M^{opt}).
		    \end{multline*}
		By setting $\varepsilon = 1$ and $\gamma = 2$, we attain the
		approximation ratio of $\frac{8}{3}k + \frac{64}{9}k' + \frac{8}{9}$ for all $k$ and $k'$. 
		\end{theorem}
		%For example, this ratio for $\varepsilon = 1$ and $\gamma = 2$ is 
		
		It is possible to obtain better ratios for fixed $k$ and $k'$ by 
		a more careful choice of the parameters $\varepsilon$ and $\gamma$. 
		% For instance, if we set $k' = 1$, when $k=2,3,$ and $4$, we have the respective ratios of $13.055, 15.283,$ and $17.325$. 
		
    \subsection{Analysis for Non-Monotone Submodular Function Maximization}
    
        In this section, we assume $\frac{1}{2 + k (1 + \varepsilon)} \leq p \leq \frac{1}{k + 1}$. 
        The following lemma is the counterpart of Lemma~\ref{lemma_eg_ef}, 
        whose proof again uses the technique of conditioning (in a more 
        general form). 
        
        \begin{lemma}
            We have
            $\left((1 + \varepsilon)k + \frac{(\gamma \cdot (1 + \varepsilon))^2}{\gamma \cdot (1 + \varepsilon) - 1} k' + \frac{1 + \varepsilon}{\varepsilon}\right)\mathbb{E}[g(S)] \geq \mathbb{E}[f(S \cup M^{opt} \,|\, \emptyset)]$.
        \end{lemma}
        
        \begin{proof}
        
        By Lemma~\ref{lemma_gSp_fSp-matchoid}, for any realization on randomness, we have  $\frac{1 + \varepsilon}{\varepsilon}g(S) \geq f(S\,|\,\emptyset)$, so the inequality also holds in expectation.
            
            We next show that, for any $e \in M^{opt}$ we have 
            \begin{equation}\mathbb{E}[f(e \,|\, S)] \leq (1 + \varepsilon)\mathbb{E}\left[\sum_{u \in e}w_u(e) + \frac{w_{\mathcal{M}}'(e)}{1 + \varepsilon}\right]. 
            \label{equ:thirdOne-matchoid}
            \end{equation}
            
            Let $e \in M^{opt}$. Conditioning on the event $A_e = [f(e \,|\, S^{(t-1)}) \leq (1 + \varepsilon)(\sum_{u \in e}w_u^*(e) + \gamma \cdot w_{\mathcal{M}}^*(e))]$ (\emph{i.e.} $e$ cannot be part of $S$), we have 
            \begin{align*}
				\mathbb{E}[f(e \,|\, S)\,|\,A_e] &\leq \mathbb{E}[f(e \,|\, S^{(t-1)})\,|\,A_e]\\
				&\leq \mathbb{E}\left[(1 + \varepsilon)(\sum_{u \in e}w_u^*(e) + \gamma \cdot w_{\mathcal{M}}^*(e))\,|\,A_e\right]\\
				&= (1 + \varepsilon)\mathbb{E}\left[\sum_{u \in e}w_u(e) + \frac{1}{1 + \varepsilon} \cdot w_{\mathcal{M}}'(e)\,|\,A_e\right].
			\end{align*}
			
			For the condition $\overline{A_e}$, recall that it means
			that with probability $p$, the edge is added into $S$ 
			and with probability $1-p$, it is not. So

			\begin{align*}
			    &\mathbb{E}\left[\sum_{u \in e}w_u(e) + \frac{1}{1 + \varepsilon} \cdot w_{\mathcal{M}}'(e)\,|\,\overline{A_e}\right]\\
				=\,&p \cdot \mathbb{E}\left[\left(k + \frac{1}{1 + \varepsilon}\right) f(e \,|\, S^{(t-1)}) - \left(k + \frac{1}{1 + \varepsilon} - 1\right) \sum_{u \in e}w_u^*(e) + k w_{\mathcal{M}}^*(e) \,|\,\overline{A_e}\right]\\
				&\quad+ (1 - p)\cdot\mathbb{E}\left[\sum_{u \in e}w_u^*(e) + \gamma \cdot w_{\mathcal{M}}^*(e)\,|\,\overline{A_e}\right]\\
				\geq\,&p \cdot \mathbb{E}\left[\left(k + \frac{1}{1 + \varepsilon}\right) f(e \,|\, S^{(t-1)}) - k\sum_{u \in e}w_u^*(e) - k \cdot \gamma \cdot w_{\mathcal{M}}^*(e) \,|\,\overline{A_e}\right]\\
				&\quad+ (1 - p)\cdot\mathbb{E}\left[\sum_{u \in e}w_u^*(e) + \gamma \cdot w_{\mathcal{M}}^*(e)\,|\,\overline{A_e}\right]\\
				=\,& \left(k + \frac{1}{1 + \varepsilon}\right) \cdot p \cdot \mathbb{E}\left[f(e \,|\, S^{(t-1)})\,|\,\overline{A_e}\right]\\
				&\quad+ (1 - (k + 1)\cdot p)\cdot\mathbb{E}\left[\sum_{u \in e}w_u^*(e) + \gamma \cdot w_{\mathcal{M}}^*(e)\,|\,\overline{A_e}\right]\\
				\geq\,&\left(k + \frac{1}{1 + \varepsilon}\right) \cdot p \cdot \mathbb{E}\left[f(e \,|\, S^{(t-1)})\,|\,\overline{A_e}\right],
			\end{align*}
			where in the second step we use the fact that $\gamma > 1$ and in the last inequality that $p \leq \frac{1}{k + 1}$.
			
			Now as $p \geq \frac{1}{2 + k (1 + \varepsilon)}$, we have 
            \begin{align*}
                (1 + \varepsilon)\mathbb{E}\left[\sum_{u \in e}w_u(e) + \frac{w_{\mathcal{M}}'(e)}{1 + \varepsilon} \,|\,\overline{A_e}\right]
                &\geq (1 + \varepsilon)\left(k + \frac{1}{1 + \varepsilon}\right)\\
                &\quad \cdot p \cdot \mathbb{E}\left[f(e \,|\, S^{(t-1)})\,|\,\overline{A_e}\right]\\
                &\geq(1 - p) \cdot \mathbb{E}\left[f(e \,|\, S^{(t-1)})\,|\,\overline{A_e}\right]\\
                &\geq\mathbb{E}\left[f(e \,|\, S)\,|\,\overline{A_e}\right],  
            \end{align*}
            and we have established~(\ref{equ:thirdOne-matchoid}). 
            
            Similar to the proof of Lemma \ref{lemma_eg_ef} we get
            \begin{align*}
            	\mathbb{E}\left[f(M^{opt} \,|\, S)\right] &\leq \sum_{e \in M^{opt}} \mathbb{E}\left[f(e \,|\, S)\right]\\
                &\leq (1 + \varepsilon) \sum_{e \in M^{opt}} \mathbb{E}\left[\sum_{u \in e} w_u(e)\right] +  \mathbb{E}[w_{\mathcal{M}}'(M^{opt})]\\	
                &\leq k(1 + \varepsilon)\mathbb{E}[g(S)] +  \mathbb{E}[w_{\mathcal{M}}'(M^{opt})].
            \end{align*}
            By Proposition~\ref{pro:matchoid} we know that for any realization of randomness, \[w_{\mathcal{M}}'(M^{opt}) \leq \frac{(\gamma \cdot (1+\varepsilon))^2}{\gamma \cdot (1+\varepsilon)-1} k' w_{\mathcal{M}}'(S_f) = \frac{(\gamma \cdot (1+\varepsilon))^2}{\gamma \cdot (1+\varepsilon)-1} k' \cdot g(S).\] Thus we get $\mathbb{E}\left[f(M^{opt} \,|\, S)\right] \leq \left((1 + \varepsilon)k + \frac{(\gamma \cdot (1+\varepsilon))^2}{\gamma \cdot (1+\varepsilon)-1} k'\right)\mathbb{E}[g(S)]$. 
            
            Now the bound on $\mathbb{E}\left[f(M^{opt} \,|\, S)\right]$ 
            and the bound on $\mathbb{E}\left[f(S \,|\, \emptyset)\right]$ (argued in the beginning of the proof) give us the lemma. 
        \end{proof}
        
        Finally, using Lemma~\ref{lemma_ineg_h} we obtain:
        \begin{theorem}
        \label{theo_submo_nonmono_matchoid}
            For non-negative submodular functions, Algorithm~\ref{streaming_part_matchoid_submod} with $p = \frac{1}{2 + k(1 + \varepsilon)}$ combined with Algorithm~\ref{greedy_part} provides a $b$-matching $M$ independent in $\mathcal{M}$ such that:
            \begin{multline*}
                \frac{2 + k(1 + \varepsilon)}{1 + k(1 + \varepsilon)} \cdot \left(1 + \frac{1}{\gamma \cdot (1 + \varepsilon) - 1}\right)
                \cdot \left((1+ \varepsilon) k + \frac{(\gamma \cdot (1+\varepsilon))^2}{\gamma \cdot (1+\varepsilon)-1} k' + 1 + \frac{1}{\varepsilon}\right)\mathbb{E}[f(M)] \\
                \geq f(M^{opt}).
            \end{multline*}
            Setting $\varepsilon = 1$ and $\gamma = 2$ we obtain the ratio of $\frac{2k + 2}{2k + 1} \cdot \left(\frac{8}{3}k + \frac{64}{9}k' + \frac{8}{3}\right)$ for all $k$ and $k'$. 
        \end{theorem}
        
        As in the last section, it is possible to obtain better ratios for fixed $k$  and $k'$ by 
		a more careful choice of the parameters $\varepsilon$ and $\gamma$.
        
\bibliography{names, library}	

\begin{thebibliography}{10}

\bibitem{Buchbinder2014a}
Niv Buchbinder, Moran Feldman, Joseph Naor, and Roy Schwartz.
\newblock Submodular maximization with cardinality constraints.
\newblock In {\em Proc\@. 25th SODA}, pages 1433--1452, 2014.

\bibitem{DBLP:journals/mp/ChakrabartiK15}
Amit Chakrabarti and Sagar Kale.
\newblock Submodular maximization meets streaming: matchings, matroids, and
  more.
\newblock {\em Math. Program.}, 154(1-2):225--247, 2015.

\bibitem{ChekuriGQ15}
Chandra Chekuri, Shalmoli Gupta, and Kent Quanrud.
\newblock Streaming algorithms for submodular function maximization.
\newblock In {\em Proc\@. 42nd ICALP}, pages 318--330, 2015.

\bibitem{CS2014}
Michael Crouch and D.M. Stubbs.
\newblock Improved streaming algorithms for weighted matching, via unweighted
  matching.
\newblock {\em Leibniz International Proceedings in Informatics, LIPIcs},
  28:96--104, 09 2014.
\newblock \href {https://doi.org/10.4230/LIPIcs.APPROX-RANDOM.2014.96}
  {\path{doi:10.4230/LIPIcs.APPROX-RANDOM.2014.96}}.

\bibitem{Epstein2009}
Leah Epstein, Asaf Levin, Julián Mestre, and Danny Segev.
\newblock Improved approximation guarantees for weighted matching in the
  semi-streaming model.
\newblock {\em SIAM Journal on Discrete Mathematics}, 25, 07 2009.
\newblock \href {https://doi.org/10.4230/LIPIcs.STACS.2010.2476}
  {\path{doi:10.4230/LIPIcs.STACS.2010.2476}}.

\bibitem{FKMSZ2005}
Joan Feigenbaum, Sampath Kannan, Andrew McGregor, Siddharth Suri, and Jian
  Zhang.
\newblock On graph problems in a semi-streaming model.
\newblock {\em Theoretical Computer Science}, 348:207--216, 12 2005.
\newblock \href {https://doi.org/10.1016/j.tcs.2005.09.013}
  {\path{doi:10.1016/j.tcs.2005.09.013}}.

\bibitem{DBLP:conf/nips/FeldmanK018}
Moran Feldman, Amin Karbasi, and Ehsan Kazemi.
\newblock Do less, get more: Streaming submodular maximization with
  subsampling.
\newblock In {\em NeurIPS}, pages 730--740, 2018.

\bibitem{Feldman2011}
Moran Feldman, Joseph~(Seffi) Naor, Roy Schwartz, and Justin Ward.
\newblock Improved approximations for k-exchange systems.
\newblock In {\em Proc\@. 19th ESA}, pages 784--798, 2011.
\newblock URL:
  \url{http://portal.acm.org/citation.cfm?id=2040572.2040658&coll=DL&dl=GUIDE&CFID=95852163&CFTOKEN=76709828}.

\bibitem{GJS2021}
Paritosh Garg, Linus Jordan, and Ola Svensson.
\newblock Semi-streaming algorithms for submodular matroid intersection.
\newblock In {\em IPCO}, 2021.

\bibitem{GW2019}
Mohsen Ghaffari and David Wajc.
\newblock {Simplified and Space-Optimal Semi-Streaming
  (2+$\varepsilon$)-Approximate Matching}.
\newblock In Jeremy~T. Fineman and Michael Mitzenmacher, editors, {\em 2nd
  Symposium on Simplicity in Algorithms (SOSA 2019)}, volume~69 of {\em
  OpenAccess Series in Informatics (OASIcs)}, pages 13:1--13:8, Dagstuhl,
  Germany, 2018. Schloss Dagstuhl--Leibniz-Zentrum fuer Informatik.
\newblock URL: \url{http://drops.dagstuhl.de/opus/volltexte/2018/10039}, \href
  {https://doi.org/10.4230/OASIcs.SOSA.2019.13}
  {\path{doi:10.4230/OASIcs.SOSA.2019.13}}.

\bibitem{HS2021}
Chien-Chung Huang and Fran\c{c}ois Sellier.
\newblock {Semi-Streaming Algorithms for Submodular Function Maximization Under
  b-Matching Constraint}.
\newblock In Mary Wootters and Laura Sanit\`{a}, editors, {\em Approximation,
  Randomization, and Combinatorial Optimization. Algorithms and Techniques
  (APPROX/RANDOM 2021)}, volume 207 of {\em Leibniz International Proceedings
  in Informatics (LIPIcs)}, pages 14:1--14:18, Dagstuhl, Germany, 2021. Schloss
  Dagstuhl -- Leibniz-Zentrum f{\"u}r Informatik.
\newblock URL: \url{https://drops.dagstuhl.de/opus/volltexte/2021/14707}, \href
  {https://doi.org/10.4230/LIPIcs.APPROX/RANDOM.2021.14}
  {\path{doi:10.4230/LIPIcs.APPROX/RANDOM.2021.14}}.

\bibitem{LW2021}
Roie Levin and David Wajc.
\newblock Streaming submodular matching meets the primal-dual method.
\newblock In {\em SODA}, pages 1914--1933, 2021.

\bibitem{Mcg2005}
Andrew McGregor.
\newblock Finding graph matchings in data streams.
\newblock In Chandra Chekuri, Klaus Jansen, Jos{\'e} D.~P. Rolim, and Luca
  Trevisan, editors, {\em Approximation, Randomization and Combinatorial
  Optimization. Algorithms and Techniques}, pages 170--181, Berlin, Heidelberg,
  2005. Springer Berlin Heidelberg.

\bibitem{Mut2005}
S.~Muthukrishnan.
\newblock Data streams: Algorithms and applications.
\newblock {\em Found. Trends Theor. Comput. Sci.}, 1(2):117–236, August 2005.
\newblock \href {https://doi.org/10.1561/0400000002}
  {\path{doi:10.1561/0400000002}}.

\bibitem{PS2017}
Ami Paz and Gregory Schwartzman.
\newblock A $(2 + \varepsilon)$-approximation for maximum weight matching in
  the semi-streaming model.
\newblock In {\em Proceedings of the 2017 Annual ACM-SIAM Symposium on Discrete
  Algorithms (SODA)}, pages 2153--2161, 2017.
\newblock URL:
  \url{https://epubs.siam.org/doi/abs/10.1137/1.9781611974782.140}, \href
  {http://arxiv.org/abs/https://epubs.siam.org/doi/pdf/10.1137/1.9781611974782.140}
  {\path{arXiv:https://epubs.siam.org/doi/pdf/10.1137/1.9781611974782.140}},
  \href {https://doi.org/10.1137/1.9781611974782.140}
  {\path{doi:10.1137/1.9781611974782.140}}.

\bibitem{Sch2003}
Alexander Schrijver.
\newblock {\em Combinatorial optimization: polyhedra and efficiency},
  volume~24.
\newblock Springer Science \& Business Media, 2003.

\bibitem{Zel2008}
Mariano Zelke.
\newblock Weighted matching in the semi-streaming model.
\newblock {\em Algorithmica}, 62:1--20, 02 2008.
\newblock \href {https://doi.org/10.1007/s00453-010-9438-5}
  {\path{doi:10.1007/s00453-010-9438-5}}.

\end{thebibliography}

\appendix
\section{Making Algorithm~\ref{streaming_part} Memory-Efficient}
    \label{appendix_memory}
		
	We explain how to guarantee the space requirement promised in Theorem~\ref{thm:first}. In this section, $w_{min}$ denotes the minimum non-zero value of the weight of an edge, and $w_{max}$ the maximum weight of an edge. 
	Moreover, we set $W = w_{max}/w_{min}$. We also define $M_{max}$ as a given maximum cardinality $b$-matching.
	
	Let $\alpha = (1 + \varepsilon) > 1$. In Algorithm~\ref{streaming_part_mem} we add an edge $e$ to $S$ only if $w(e) > \alpha \sum_{u \in e} w_u^*(e)$ (Line 7). For the moment we set $d = 0$ in our analysis, ignoring Lines~14-16 of the algorithm.
	
	\begin{algorithm}
	\caption{Streaming phase for weighted matching, memory-efficient}\label{streaming_part_mem}
	\begin{algorithmic}[1]
	\State $S \gets \emptyset$
	\State $\forall v \in V : Q_v \gets (Q_{v,1} = \emptyset, \cdots, Q_{v, b_v} = \emptyset)$
	\For{$e = e_t,\, 1 \leq t \leq |E|$ an edge from the stream}
	
		\For{$u \in e$}
			\State $w_u^*(e) \gets \min \{w_u(Q_{u, q}.top()) : 1 \leq q \leq b_u\}$
			\State $q_u(e) \gets q \text{ such that $w_u(Q_{u, q}.top()) = w_u^*(e)$}$
		\EndFor
	
		\If {$w(e) > \alpha \sum_{u \in e} w_u^*(e)$} \Comment stricter condition here
			\State $g(e) \gets w(e) - \sum_{u \in e} w_u^*(e)$
			\State $S \gets S \cup \{e\}$
			\For{$u \in e$}
				\State $w_u(e) \gets w_u^*(e) + g(e)$
				\State $r_u(e) \gets Q_{u, q_u(e)}.top()$
				\State $Q_{u, q_u(e)}.push(e)$
				\If{$d = 1 \text{ and } Q_{u, q_u(e)}.length() > \beta$} \Comment remove some small element
				    \State let $e'$ be the $\beta + 1$-th element from the top of $Q_{u, q_u(e)}$
					\State mark $e'$ as erasable, so that when it will no longer be on the top of any queue, it will be removed from $S$ and from all the queues it appears in
				\EndIf
		    \EndFor
		\EndIf
	\EndFor
	\end{algorithmic}
	\end{algorithm}
	
	\begin{lemma}
		The set $S$ obtained at the end of algorithm \ref{streaming_part_mem} when $d = 0$ guarantees that $2\alpha g(S) \geq w(M^{opt})$.
	\end{lemma}
	
	\begin{proof}
		We proceed as in~\cite{GJS2021}. Let $w_{\alpha} : E \rightarrow \mathbb{R}$ such that $w_{\alpha}(e) = w(e)$ for $e \in S$ and $w_{\alpha}(e) = \frac{w(e)}{\alpha}$ for $e \in E \backslash S$. We can observe that with the weights $w_{\alpha}$, Algorithm~\ref{streaming_part} gives the same set $S$ as Algorithm~\ref{streaming_part_mem} with the weights $w$. We deduce that $w_{\alpha}(M^{opt}) \leq 2 g(S)$ and then, as $w \leq \alpha w_{\alpha}$, we get $2\alpha g(S) \geq w(M^{opt})$.
	\end{proof}
	
	Hence, using same arguments as in~\cite{GW2019,LW2021} we obtain the following: 
	\begin{theorem}
		Algorithm~\ref{streaming_part_mem} (with $d = 0$) combined with Algorithm~\ref{greedy_part} gives a $2 + \varepsilon$ approximation algorithm by using $O\left(\log_{1+\varepsilon} (W/\varepsilon) \cdot |M_{max}|\right)$ variables.
		%, giving a semi-streaming algorithm when $W$ is bounded.
	\end{theorem}
	
	\begin{proof}
	    In a given queue, the minimum non-zero value that can be attained is $\frac{\varepsilon}{1 + \varepsilon}w_{min}$ and the maximum value that can be attained is $w_{max}$. As the value of the top element of the queue increases at least by a factor $1 + \varepsilon$ for each inserted element, a given queue contains at most $\log_{1 + \varepsilon}(W/\varepsilon) + 1$ edges. Hence, a vertex $v\in V$ contains at most $\min\{|\delta(v)|, b_v \cdot (\log_{1 + \varepsilon}(W/\varepsilon) + 1)\}$ elements of $S$ at the end of the algorithm.
	    
	    Then let $U \subseteq V$ be the set of \emph{saturated} vertices of $V$ by $M_{max}$, \emph{i.e.} the set of vertices $v \in V$ such that $|\delta(v) \cap M_{max}| = \min\{|\delta(v)|, b_v\}$. By construction, $U$ covers all the edges in $E \backslash M$ and $\sum_{v \in U}\min\{|\delta(v)|, b_v\} \leq 2 |M_{max}|$. As all edges of $S$ are either in $M_{max}$ or have at least one endpoint in $U$, we get:
	    \begin{align*}
	        |S| &\leq  |M_{max}| + \sum_{v \in U} |\delta(v) \cap S| \leq |M_{max}| + \sum_{v \in U} \min\{|\delta(v)|, b_v \cdot (\log_{1 + \varepsilon}(W/\varepsilon) + 1)\}\\
	        &\leq |M_{max}| + \sum_{v \in U} \min\{|\delta(v)|, b_v\} \cdot (\log_{1 + \varepsilon}(W/\varepsilon) + 1) \leq (2\log_{1 + \varepsilon}(W/\varepsilon) + 3) \cdot |M_{max}|,
	    \end{align*}
	    so the memory consumption of the algorithm is $O(\log_{1 + \varepsilon}(W/\varepsilon) \cdot |M_{max}|)$
	\end{proof}
	
	Modifying an idea from Ghaffari and Wajc~\cite{GW2019} we can further improve the memory consumption of the algorithm, especially if $W$ is not bounded. For that, set $\beta = 1 + \frac{2 \log (1/\varepsilon)}{\log (1 + \varepsilon)} = 1 + \log_{1 + \varepsilon}(1/\varepsilon^2)$ in Algorithm~\ref{streaming_part_mem} as the maximum size of a queue, and set $d = 1$ (so that we  now consider the whole code).
	
	Consider an endpoint $u$ of the newly-inserted edge $e$. If the queue $Q_{u,q_u(e)}$ becomes too long (more than $\beta$ elements), it means the gain $g(e')$ of the $\beta + 1$-th element from the top of the queue (we will call that element $e'$) is very small compared to $g(e)$, so we then can ``potentially'' remove $e'$ from $S$ and from the queues without hurting too much $g(S)$. In the code, we will mark this edge $e'$ as \emph{erasable}, so that when $e'$ will no longer be on top of any queue, it will be removed from $S$ and all the queues it appears in. To be able to do these eviction operations, the queues have to be implemented with doubly linked lists.

	If an edge $e = \{u, v\}$ is marked as erasable by Algorithm~\ref{streaming_part_mem} ($d = 1$) because an edge $e' = \{u, v'\}$ is added to $S$, then we say that $e'$ \emph{evicted} $e$ (and that 
	$e$ \emph{was evicted by} $e'$).

	\begin{lemma}
		If $e = \{u, v\}$ is evicted by $e' = \{u, v'\}$, then $g(e') \geq g(e)/\varepsilon$.
	\end{lemma}
	
	\begin{proof}
		We have $g(e') \geq \varepsilon (w_u^*(e') + w_{v'}^*(e')) \geq \varepsilon w_u^*(e') \geq \varepsilon(1 + \varepsilon)^{\beta - 1}g(e) \geq g(e)/\varepsilon$ because after $e = e_t$ is added to $Q_{u, q_u(e)}$ we have $w_u(Q_{u, q_u(e)}^{(t)}) \geq g(e)$ and each time an element is added to $Q_{u, q_u(e)}$ the value $w_u(Q_{u, q_u(e)})$ is multiplied at least by $(1 + \varepsilon)$.
	\end{proof}
	
	\begin{theorem}
		For $\varepsilon \leq \frac{1}{4}$, Algorithm~\ref{streaming_part_mem} with $d = 1$ combined with Algorithm~\ref{greedy_part} gives a $ 2 + \varepsilon$ approximation algorithm by using $O\left(\log_{1 + \varepsilon} (1/\varepsilon)\cdot|M_{max}| + \sum_{v \in V} b_v\right)$ variables.
	\end{theorem}
	
	\begin{proof}
		For an element $e$ that was not evicted from $S$ in Algorithm~\ref{streaming_part_mem}, denote by $\mathcal{E}_e$ the elements that were evicted by $e$ directly or indirectly (in a chain of evictions). This set $\mathcal{E}_e$ contains at most $2$ elements that were directly evicted when $e$ was inserted in $S$, and their associated gain is at most $\varepsilon g(e)$ for each, and at most $4$ elements indirectly evicted by $e$ when these $2$ evicted elements were inserted in $S$, and their associated gain is at most equal to $\varepsilon^2 g(e)$ for each, and so on. Then, as $\varepsilon \leq \frac{1}{4}$, 
		\[\sum_{e' \in \mathcal{E}_e} g(e') \leq \sum_{i = 1}^{\infty} (2 \varepsilon)^i g(e) \leq 2 \varepsilon g(e) \sum_{i = 0}^{\infty}(1/2)^i = 4 \varepsilon g(e)\]
		
		Therefore, if $S_0$ denotes the set $S$ obtained by Algorithm~\ref{streaming_part_mem} when $d = 0$ and $S_1$ denotes the set $S$ obtained by Algorithm~\ref{streaming_part_mem} when $d = 1$, we get:
		\[g(S_0) - g(S_1) \leq 4 \varepsilon g(S_1)\]
		because the elements inserted in $S$ are exactly the same for the two algorithms, the only difference being that some elements are missing in $S_1$ (but these elements were removed when they no longer had any influence on the values of $w^*$ and thereby no influence on the choice of the elements inserted in $S$ afterwards).
		We have then:
		\[w(M^{opt}) \leq 2(1 + \varepsilon) g(S_0) \leq 2(1 + \varepsilon)(1 + 4 \varepsilon) g(S_1) \leq 2(1 + 6\varepsilon) g(S_1)\]
		and Algorithm~\ref{greedy_part} will provide a $2(1 + 6 \varepsilon)$ approximation of the optimal $b$-matching. 
		In fact, the analysis of Algorithm~\ref{greedy_part} with $S_1$ as input is almost the same as in the proof of Lemma~\ref{lemma_greedy_weight}, the only difference being that now the weight of an element is no longer necessarily equal to the sum of the gains of elements below it but can be higher (which is not an issue).
		
		Regarding the memory consumption of the algorithm, one can notice that, using the same notation $U$ for the set of the elements saturated by $M_{max}$ as previously, and because there are only up to $\sum_{v \in V}b_v$ elements on top of the queues that cannot be deleted (and thus these edges are the ones making some queues in $U$ exceed the limit $\beta$), we obtain:
		\[|S| \leq |M_{max}| + \sum_{v \in V}b_v + \sum_{v \in U} \beta \cdot \min\{|\delta(v)|, b_v\} \leq \sum_{v \in V}b_v + (2\beta + 1)\cdot |M_{max}|,\]
		and therefore the algorithm uses $O\left(\sum_{v \in V}b_v + \log_{1 + \varepsilon}(1/\varepsilon)\cdot |M_{max}|\right)$ variables.
	\end{proof}
	
	\begin{remark}
	    Ideas presented here for the maximum weight $b$-matching can be easily extended to submodular function maximization, for hypergraphs under a matroid constraint or under matchoid constraints, namely for the algorithms presented in Sections~\ref{sec:submodular}, \ref{sec:matroid}, and~\ref{sec:matchoid}.
	\end{remark}
	
\section{Example of Different Behavior Compared to~\cite{LW2021}}
    \label{appendix_difference}

    Here is an example to show the difference of behavior between our algorithm and the one proposed in~\cite{LW2021}. Consider a set of four vertices $V = \{v_1, v_2, v_3, v_4\}$. We set $b_{v_1} = 2$ and $b_{v_i} = 1$ for $2 \leq i \leq 4$. Let $E = \{e_1, e_2, e_3\}$ with $e_1 = \{v_1, v_2\}$ and $w(e_1) = 2$, $e_2 = \{v_1, v_3\}$ and $w(e_2) = 7$, $e_3 = \{v_1, e_4\}$ and $w(e_3) = 4$. Using only one dual variable for each vertex, the algorithm of Levin and Wajc~\cite{LW2021} takes $e_1$ and $e_2$ but discards $e_3$ because the dual variable $\phi_{v_1}$ is equal to $4$ when $e_3$ is processed. On the other hand, our algorithm, when processing $e_3$, compares $w(e_3)$ with $w_{v_1}(e_1) = 2$. Therefore, $e_3$ is added into $S$ and our algorithm provides a $b$-matching of weight $11$ instead of $9$.

\end{document}